\documentclass{elsart}

\usepackage{amssymb}
\usepackage{amsmath}
\usepackage{verbatim}
\usepackage{epsfig}
\usepackage{color}
\definecolor{Blue}{rgb}{1,0,0}
\usepackage{ulem}
\usepackage{cancel}
\usepackage{bbm}
\usepackage{relsize}
\usepackage{tikz}

\newcommand{\lb}{\mathopen{[\![ }}
\newcommand{\rb}{\mathclose{]\!]}}

\makeatletter
\newcommand{\subalign}[1]{%
  \vcenter{%
    \Let@ \restore@math@cr \default@tag
    \baselineskip\fontdimen10 \scriptfont\tw@
    \advance\baselineskip\fontdimen12 \scriptfont\tw@
    \lineskip\thr@@\fontdimen8 \scriptfont\thr@@
    \lineskiplimit\lineskip
    \ialign{\hfil$\m@th\scriptstyle##$&$\m@th\scriptstyle{}##$\crcr
      #1\crcr
    }%
  }
}
\makeatother

\DeclareMathOperator{\seq}{\textbf{\textup{Seq}}}
\DeclareMathOperator{\para}{\textbf{\textup{Para}} }
\DeclareMathOperator{\width}{width}
\DeclareMathOperator{\depth}{depth}
\DeclareMathOperator{\supp}{supp}

\DeclareMathOperator{\id}{\textbf{{Id}}}
\DeclareMathOperator{\pom}{\textbf{{Pom}}}

\newcommand{\llbracket}{\mathopen{[\![ }}
\newcommand{\rrbracket}{\mathclose{]\!]}}

\begin{document}

\begin{frontmatter}

\title{Completeness Theorems for Pomset Languages and Concurrent Kleene Algebras}

\author[shef]{Michael R Laurence}
\author[shef]{Georg Struth}
\address[shef]{ Department of Computer Science, University of Sheffield}

\date{\today}

\begin{abstract}
Pomsets constitute one of the most basic models of concurrency. A pomset is a generalisation of a word over an alphabet in that letters may be partially ordered rather than totally ordered.  A  term $ t $ using the bi-Kleene operations   $0,1, +,  \cdot\, ,^*, \parallel, ^{(*)}$ 
 defines a set $ \mathopen{[\![ } t \mathclose{]\!] } $ of pomsets in a natural way. 
  We prove that every valid universal equality over pomset languages using these operations is a consequence of the equational theory of regular languages (in which parallel multiplication and iteration are undefined) plus that of the commutative-regular languages (in which sequential multiplication and iteration are undefined). We also show that the class of \textit{rational} pomset languages (that is, those languages generated from singleton pomsets  using the bi-Kleene operations) is closed under all Boolean operations. 
  
   An \textit{ideal} of a pomset $p$ is a pomset using the letters of $p$, but having an ordering at least as strict as $p$. A bi-Kleene term $t$ thus defines 
 the set  $\id (\mathopen{[\![ } t \mathclose{]\!] }) $ of ideals of pomsets in  $ \mathopen{[\![ } t \mathclose{]\!] } $. We prove that if  $t$ does not contain commutative iteration $^{(*)}$ (in our terminology, $t$ is bw-rational) then 
  $\id (\mathopen{[\![ } t \mathclose{]\!] }) \cap \pom_{sp}$, where $ \pom_{sp}$ is the set of  pomsets generated from singleton pomsets using sequential and parallel multiplication ($ \cdot$ and $ \parallel$) is defined by a bw-rational term, and 
 if  two such terms $t,t'$  define the same ideal language, then $t'=t$ is provable from the Kleene  axioms for  $0,1, +,  \cdot\, ,^*$  plus the commutative idempotent semiring axioms for  $0,1, +, \parallel$  plus the exchange law 
 $ (u \parallel v)\cdot ( x \parallel y) \le    v \cdot y   \parallel u \cdot x $.  
\end{abstract}
 
\end{frontmatter}

\def\gr{\mathit{Flowchart}}
\def\edgtyp{\mathit{edgeType}}
\def\labs{\mathit{Labels}}

\section{Introduction}

Pomsets may be regarded as a generalisation of both words over an alphabet and commutative words over an alphabet as studied by Conway \cite[Chapter 11]{ConwayJH:71:regafm}. Words of the former kind are generated using sequential multiplication $(\cdot)$, whereas commutative words
are generated using parallel multiplication $( \parallel )$. Both operations are defined on the set of pomsets.  Pomsets have been widely used to model the behaviour of concurrent systems \cite{Pratt:82:OCP,Pratt:85:COTM,Brookes:02:TPFFA,Gaston:Mislove:TCS:2002:pomset,Zhao:10:TPSSV}.

A pomset over an alphabet $ \Sigma$ is defined by a finite labelled partially ordered set; that is, a finite partially ordered set (or poset) $V$ on which a labelling function into $ \Sigma$ is defined. Since the focus is on the labelling rather than the elements of $V$, isomorphic labelled posets are regarded as defining the same pomset. For pomsets $p_1,p_2$ defined by posets $V_1,V_2$, the sequential and parallel products $p_1 \cdot p_2$ and $p_1 \parallel p_2$ are defined, respectively, by placing  the elements of $V_1$ below those of $V_2$, and placing  the elements of $V_1$ and  $V_2$ side by side. 

Given any monoid $(M, \cdot ,1)$, the operation $ \cdot $ can be extended pointwise to the  power set $2^M$ of $M$, and if  the regular operations $0,1,+,  \cdot , ^*$ are defined in the usual way for $2^M$ (in particular, $ P^* = \cup_{i \ge 0} P^i$), then the algebra thus defined is an example of a Kleene algebra (Definition \ref{kleene.alg.defn}). 
Since the set of pomsets over an alphabet $\Sigma$ is a monoid with respect to the operations $\cdot, 1$ and a commutative monoid with respect to $ \parallel, 1$, 
 the class of languages (sets) of pomsets over  $\Sigma$
 is thus  a bi-Kleene algebra with respect to the bi-Kleene  operations $0,1,+, \cdot, ^* \parallel, ^{(*)} $, where parallel iteration $^{(*)}$ is defined analogously to $^*$, but using parallel multiplication. 
 A pomset language is \textit{rational} if is defined by a bi-Kleene term over an alphabet $\Sigma$. 
 This is a simplification of the phrase series-parallel-rational used by Lodaya and Weil \cite{Lodaya:weil:98:SPPAAl:27,Lodaya98:weil:98:SPLBW}. If $t$ is a bi-Kleene term, then we use $ \lb t  \rb$ to denote the language that it defines.

In this paper we prove the following theorems, for bi-Kleene terms $t,t'$ over an alphabet $ \Sigma$; 

\begin{itemize}
\item
The language $ \lb t \rb - \lb t' \rb$ is rational.
\item
It is decidable whether  $ \lb t \rb = \lb t' \rb$ holds.
\item
 If $ \lb t \rb = \lb t' \rb$ holds, then $t=t'$ holds in every bi-Kleene algebra. Equivalently, the algebra of pomset languages generated by the bi-Kleene operations from the singleton pomsets with label in $ \Sigma$ is the free bi-Kleene algebra with basis $ \Sigma$. 
\end{itemize}

This latter theorem is, in effect, a strengthening of Gischer~\cite[Theorem 4.3]{Gischer:1988:EqThPo}, in which neither of the two Kleene stars $^*,\, ^{(*)}$ was considered. 
Bi-Kleene algebras have been proposed as tools  for the verification of concurrent programs \cite{CKA}. Our completeness and decidability results can make reasoning about such programs simpler and less problematic.

\subsection{New theorems for pomset ideals and bw-rational operations}

Given a pomset $p$, an \textit{ideal} of $p$ is a pomset  that may be represented using the same vertex set as $p$, with the same labelling, but whose partial ordering is at least as strict as that for $p$. We write $\id(L)$ for a pomset language $L$ to denote the set of ideals of elements of $L$. The function $\id$ was first defined by Grabowski \cite{Grabowski}, who associated  pomset ideals (that is, pomset languages closed under $\id$) with a reachability condition between markings of a Petri net.

 The class of pomset ideals  over $\Sigma$ is a Kleene algebra with respect to the Kleene operations, but is not a Kleene algebra with respect to the commutative Kleene operations $0,1,+, \parallel, ^{(*)} $, since $ L \parallel L'$ is not an ideal if $L,L' \nsubseteq \{ 1\}$, but it can be made into a bi-Kleene  algebra if $ \parallel $ is interpreted as $(L, L') \mapsto \id(L \parallel L')$ and parallel iteration $ ^{(*)}$ is defined analogously.
 Additionally, the class of pomset ideals satisfies the \textit{{exchange law}}:
 \begin{equation}
 (u \parallel v)\cdot ( x \parallel y) \le    v \cdot y   \parallel u \cdot x  \label{exch.eqn}
\end{equation}     
where we use the abbreviation 
 \begin{equation}
  t \le t'  \overset{\text{defn}}{ \iff} t +t' =t'. \label{prec.abbrev.eqn}
\end{equation} 
 We have failed to prove an analogous result for ideals to the freeness theorem given for pomset languages above, but by abandoning the  parallel iteration operation $ ^{(*)}$ we have the following partial results. 
 We will refer to   $0,1,+, \cdot, ^* \parallel $ as \textit{bw-rational} operations (`bw' meaning bounded width) and we call algebras over  the  bw-rational  operations
that satisfy both the Kleene axioms for $0,1,+, \cdot, ^*$    and the idempotent commutative semiring axioms for   $0,1,+,  \parallel$ bw-rational algebras 
 and refer to a term  in the  bw-rational operations as a bw-rational term. 
We say that a pomset is \textit{series-parallel} if it is generated from the set of singleton pomsets using only sequential and parallel multiplication, and use $ \pom_{sp}$ to denote the set of series-parallel pomsets.  
 With these definitions, we prove for bw-rational terms $t,t'$ over an alphabet $ \Sigma$ that

\begin{itemize} 
\item
 the  language $\id( \lb t \rb) \cap \pom_{sp}$ is representable by a  bw-rational term, and  
\item
 suppose that $\id( \lb t \rb )=\id( \lb t' \rb)$, or equivalently $\id( \lb t \rb ) \cap \pom_{sp}   =\id( \lb t' \rb)  \cap \pom_{sp}$. 
 Then $t=t'$ is a consequence of the bw-rational axioms plus the exchange law (\ref{exch.eqn}). Hence the algebra of pomset ideals generated by the bi-Kleene operations from the singleton pomsets with labels in $ \Sigma$ is the free  algebra with basis $ \Sigma$ with respect to the class of bw-rational algebras satisfying the exchange law. 
\end{itemize}

This freeness result is, in effect, a generalisation of Gischer \cite[Theorem 5.9]{Gischer:1988:EqThPo}, which gave the analogous result for idempotent bi-semirings, in which the Kleene star $^*$ was not considered.  
 

\subsection{Organisation of the paper}

In Section \ref{sect.main.defns}, we give most of the basic definitions and results that will be used throughout the paper. In Section \ref{sect.bool.closed.pomset.sp}, we prove our first main theorem for rational pomset languages; in particular, we show that if $L,L'$ are rational languages, then so is $L \setminus L'$. We also show that a bi-Kleene term defining  $L \setminus L'$ can be computed from terms defining $L$ and $L'$.  In Section \ref{sect.biKleene.terms.equal.imply.langs}, we prove our second main theorem; that if two bi-Kleene terms define the same rational language, then they define the same element of every bi-Kleene algebra. 
In Section \ref{sect.main.thms.ideals}, we give further definitions for pomset ideals. We also prove that  
the set of pomset ideals defines a bi-Kleene algebra, provided that the operations $ \parallel, \, ^{(*)}$
are suitably modified. Section \ref{subsect.summary.proof.ideal.thms} gives a summary of the method of proof of our remaining theorems, which occupies Sections \ref{sect.auto.lems}--\ref{sect.main.thms.ideals}. 
In Section \ref{sect.conclusions} we give our conclusions.

\section{Kleene algebra and pomset definitions}

\label{sect.main.defns}

\begin{defn}[bi-Kleene algebras and bw-rational algebras] \rm \label{kleene.alg.defn}
A monoid, as usual, is an algebra with an associative binary operation $ \cdot$ and identity $1$. 
A bimonoid is an algebra with operations $ \cdot, \parallel, 1$ that is a monoid with respect to  $ \cdot,1$ and a commutative monoid with respect to $ \parallel, 1$. 

A Kleene algebra  is an algebra $K$ with constants $0,1$, a binary addition operation $+$, a  multiplication operation $ \cdot$ (usually omitted) and a unary iteration operation $^*$, such that the following hold; $(K, 1, \cdot) $ is a monoid, $(K,0,+)$ is a  commutative monoid  and also, for all $ x,y,z \in K$, 
\begin{align}
  & x+x =x, \qquad x(y + z) = xy + xz, \qquad (y+z) x = yx + yz , \label{dioid.eqn}
\\
& 1 + xx^* = 1 + x^*x = x^*,  \label{eqn.kleene.star.add}  
 \\
 & xy \le y \Rightarrow  x^*y \le y, \qquad yx \le y \Rightarrow  yx^* \le y,  \label{induct.eqn}
\end{align}
where (\ref{prec.abbrev.eqn}) is assumed. 
The identities (\ref{induct.eqn}) are normally called the induction axioms. The identities in  (\ref{dioid.eqn})  together with the preceding conditions amount to stating that $K$ is an idempotent semiring, or dioid.  We say that $K$ is a commutative Kleene algebra if $ \cdot$ is commutative. 

A bi-Kleene algebra is an algebra with  operations $0,1, +, \cdot, ^*, \parallel, ^{(*)}$ that is a Kleene algebra with respect to    $0,1, +, \cdot, ^* $ and a commutative Kleene algebra with respect to  $0,1, +,  \parallel, ^{(*)}$, with $ \parallel$ and $ ^{(*)}$ playing the role of $\cdot$ and $^*$ respectively in the Kleene axioms given above.  For the purposes of this paper,  we need to define \textit{bw-rational} algebras, which have operations $0,1, +, \cdot, ^*, \parallel, $ and satisfy only the conditions on the definition of a bi-Kleene algebra given above that do not mention $^{(*)}$; thus, a bw-rational algebra is a  Kleene algebra with respect to   $0,1, +, \cdot, ^* $  and is
 a commutative idempotent semiring with respect to the operations $0,1, + , \parallel$; that is, it
satisfies (\ref{dioid.eqn}) with $ \cdot$ replaced by $ \parallel$ and is a commutative monoid with respect to  $1,  \parallel$.

Given a set $ \Sigma$, we use $T_{Reg}(\Sigma)$,  $T_{ComReg}(\Sigma)$, $T_{bimonoid}(\Sigma)$, $T_{bi-KA}(\Sigma)$, and $T_{bw-Rat}(\Sigma)$ to denote the sets of terms generated from $ \Sigma$ using, respectively, the regular operations $0,1, +, \cdot, ^* $, the commutative-regular operations  $0,1, +,  \parallel, ^{(*)}$, the bimonoid operations $1, \cdot, \parallel$, the bi-Kleene operations 
$0,1, +, \cdot, ^*, \parallel, ^{(*)}$ and the bw-rational operations $0,1, +, \cdot, ^*, \parallel $.
 \end{defn}

An important class of naturally arising Kleene algebras is given by Proposition \ref{prop.powerset.kleene.natural}.

\begin{prop}[Kleene algebras defined on power sets of monoids]   \label{prop.powerset.kleene.natural} \rm
Let 
\\
$(M, 1, \cdot)$ be a monoid. Then $(2^M, 0,1, +, \cdot, ^*)$, with $0$ defining $ \emptyset$, $1$ defining  $\{  1 \}$, $+$ defining union,  $ \cdot$ given by pointwise multiplication and $S^* \overset{\text{defn}}{=} \cup_{i \ge 0} S^i$, is a Kleene algebra. 
\end{prop}

\begin{proof}
Straightforward. \qed
\end{proof}

\begin{defn}[commutative words] \rm
A commutative word  over an alphabet $ \Sigma $ is a multiset over $ \Sigma $; that is, a function from $ \Sigma $ into the set of non-negative integers. A commuting word may be represented by a word $ \sigma_1 \parallel \cdots \parallel \sigma_m $ with each    $ \sigma_i \in  \Sigma$, with two such words representing the same commutative word if and only if for each  $ \sigma \in \Sigma$, they contain the same number of occurrences of $ \sigma$. Thus the set of commutative words forms  a commutative monoid with $ \parallel$ as multiplication and the empty word $1$ as identity.  
\end{defn}

It follows from Proposition \ref{prop.powerset.kleene.natural} that the set of languages of strings over an alphabet $ \Sigma$ is a Kleene algebra, and the set of languages of commutative words over $ \Sigma$ is a commutative Kleene algebra with respect to the commutative-regular operations  $0,1, +,  \parallel, ^{(*)}$, when these are interpreted as given in the Proposition; in particular, $S^{(*)}=  \cup_{i \ge 0} S^{(i)}$, where we define 
\begin{equation}
S^{(0)} = 1, \qquad  S^{(1)} = S , \qquad
S^{(2)} = S \parallel S, \qquad S^{(3)} = S \parallel S \parallel S, \ldots \label{power.exp.eqn}
\end{equation}

\begin{defn}[pomsets and the $ \supp$ function] \rm
A labelled partial order is a $3$-tuple $(V, \le, \mu)  $, where $V$ is a set of vertices, $ \le $ is a partial ordering on the set $V$ and $ \mu : V \to \Sigma $ for an alphabet $ \Sigma$  is a labelling function. 
Two labelled partial orders  $(V, \le, \mu)  $ and $(V', \le', \mu')  $ are isomorphic if there is a bijection $ \tau : V \to V' $ that preserves ordering and labelling; that is, for $v,w \in V$, $ v \le w \iff  \tau(v) \le' \tau (w)$ and $ \mu(v) = \mu'( \tau(v))$ holds. 
A pomset is an isomorphism class of finite labelled partial orders, and a set of pomsets is usually called a language. We write $\pom( \Sigma)$ to denote the set of all pomsets with labels in an alphabet $ \Sigma$.
If $p$ is a pomset, then $ \supp(p)$ is the set of labels occurring in $p$, and if $L$ is a pomset language, then we define  $ \supp(L)= \cup_{p \in L} \supp(p)$. 
\end{defn}

Observe that a pomset whose ordering $ \le$  is total is simply a word, in the usual sense, over its labelling alphabet $ \Sigma $. Thus the word $ \sigma $ of length one for $ \sigma \in \Sigma$  is  
 the pomset with a single vertex having label $ \sigma$. On the other hand, a pomset over  $ \Sigma $  whose order relation is empty is, in effect, a commutative word $ \sigma_1 \parallel \ldots \parallel \sigma_m $ with each    $ \sigma_i \in  \Sigma$.

\begin{defn}[sequential and parallel multiplication of pomsets] \rm
 For pomsets $ p_1,p_2$ represented by the $3$-tuples $(V_1, \le_1, \mu_1)  $ and $(V_2, \le_2, \mu_2)  $ respectively, their sequential product $p_1  \cdot p_2$ and parallel product $ p_1 \parallel p_2$ are given as follows;  these definitions can easily be shown to be well-defined; that is, independent of the choice of representative $3$-tuple of each pomset $p_i$.
\begin{itemize}
\item 
 $p_1 \cdot p_2$ (usually written simply $p_1 p_2$) is represented by the $3$-tuple $(V_1 \cup V_2,  \ll , \mu )$, where the function $\mu$ agrees with each function $ \mu_i$ on the set $V_i$ and $ v \ll w$ holds if and only if either both vertices $v,w$ lie in one set $V_i$ for $i \in \{ 1,2 \}$ and $ v \le_i w$ holds,  or $ v \in V_1$ and $ w \in V_2$. 
\item
the pomset $p_1 \parallel p_2$  is represented by the $3$-tuple $(V_1 \cup V_2, \prec , \eta )$, where the function $\eta$ agrees with each function $ \mu_i$ on the set $V_i$ and $ v \prec w$ holds if and only if  both vertices $v,w$ lie in one set $V_i$ for $i \in \{ 1,2 \}$.
\end{itemize}
\end{defn}

\subsection{The bi-Kleene algebra of pomset languages}

It follows from Proposition \ref{prop.powerset.kleene.natural} that the set of pomset languages over $ \Sigma $ is a bi-Kleene algebra when equipped with the constant operations $0,1$, the operations  $ +,\,  \cdot,\, \parallel $ of arity two and the operations $ ^*  $ and  
 $^{(*)}$ of arity one, with interpretations as given in the Proposition; in particular,  
$1$ denotes the singleton  containing  the empty pomset, also denoted by  $ 1$, whose vertex set is empty, 
 the sequential and parallel products of pomset languages are defined from those of pomsets by pointwise extension, and 
for a  pomset language $P$, we define $P^* = \bigcup_{i \ge 0} P^i $ and $P^{(*)} =  \bigcup_{i \ge 0} P^{(i)}$, where $P^{(i)}$ is defined as indicated in (\ref{power.exp.eqn}).

\subsection{Series-parallel pomsets and rational pomset languages}

For an alphabet $ \Sigma$ and $ \sigma \in \Sigma$, we use $\sigma$ to refer to   the  pomset having only one vertex with label  $  \sigma$, and  for any $t \in T_{bi-KA}( \Sigma)$, we write $ \lb t \rb $ to denote
the pomset language defined by $t$, with operations interpreted as above.  Thus if 
$t \in T_{Reg}( \Sigma)$ then $\lb t \rb $ is regular; by analogy, if $t \in T_{ComReg}( \Sigma)$ then we say that  $\lb t \rb $ is commutative-regular. 
If a pomset $p$ satisfies $ \{p \}  = \lb t \rb $ for $t \in T_{bimonoid}( \Sigma)$, then we say that $p$ is a \textit{series-parallel} pomset. We write $ \pom_{sp}$ and $ \pom_{sp}(\Sigma)$ to denote, respectively, the set of all series-parallel pomsets and the set of all series-parallel pomsets with labels in $ \Sigma$. Fig. \ref{4602f.N.shape.pomset.grisly.fig} gives an example of a pomset that does not lie in $ \pom_{sp}$.

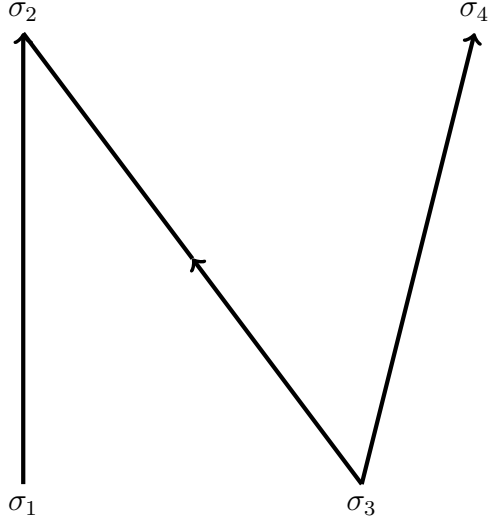
\begin{figure}[t]
\begin{center}
\begin{tikzpicture} [scale= 3]

\draw [->] [ultra thick]  (0,0)   node[below]{$ \sigma_1 $} --(0,2) node[above]{$ \sigma_2 $}   ;
\draw [->] [ultra thick] (1.5,0) --(0.75,1)   ;
\draw [ultra thick] (0.75,1) -- (0,2)  ;
\draw [->] [ultra thick] (1.5,0)  node[below]{$ \sigma_3 $} --(2,2) node[above]{$ \sigma_4 $}   ;

\end{tikzpicture}
\end{center}
\caption{An example of a pomset that is not series-parallel.} \label{4602f.N.shape.pomset.grisly.fig}
\end{figure}
We say that a pomset language $L$ is \textit{rational} if $ L = \lb t \rb $ for $t \in T_{bi-KA}( \Sigma)$; if $t \in T_{bw-Rat}( \Sigma)$, we say that $L$ is bw-rational. 
The following freeness results for the algebras of regular and commutative-regular languages have been proved.

\begin{thm} \label{kozen.390o.freeness.reg.comreg.thm}
Let $ \Sigma$ be an alphabet. If $t,t' \in T_{Reg}(\Sigma)$ and $ \lb t \rb = \lb t' \rb $ holds, then $t=t'$ holds in every Kleene algebra. If instead, $t,t' \in T_{ComReg}(\Sigma)$ and $ \lb t \rb = \lb t' \rb $ holds, then $t=t'$ holds in every commutative Kleene algebra.
\end{thm}

\begin{proof}
The assertion for regular languages was proved by Kozen \cite{kozen:complete.kleene.reg}. For commutative-regular languages, the result is implicit in the work of Conway \cite[chap.11]{ConwayJH:71:regafm}. \qed
\end{proof}

\begin{defn}[parallel and sequential pomset languages] \rm 
A pomset $p$ is 
\\
$$\begin{cases}
\text{sequential} &  \text{if } p= q_1 q_2 \\
\text{parallel} & \text{if } p= q_1 \parallel q_2
\end{cases}$$
for pomsets $q_1,q_2$ with each $q_i \not= 1$ in each case. A pomset language $L$ is sequential if every element of $L$ is sequential and non-sequential if none of its elements are sequential; we define a language to be parallel analogously. We define $\seq$ and $\para $ to be the language of all sequential and parallel  pomsets  respectively. Further, for any $ i \ge 1$ we define the language $\para_i =  \{ q_1  \parallel \cdots \parallel q_i \vert\, \text{each } q_i\not= 1  \text{ and not parallel} \}$. 
Thus 
$$ \para = \cup_{i \ge 2} \para_i $$ 
holds. 
\end{defn}

\subsection{The bi-Kleene algebra of rational pomset languages is free with respect to bi-Kleene algebras defined by power sets of bimonoids}

Lemma \ref{flamingo.5925g.lem} shows that a pomset cannot be both sequential and parallel, and hence a sequential pomset language and a parallel pomset language do not intersect.  
 
 \begin{lem} \label{flamingo.5925g.lem}
Let $p_1, p_2,q_1,q_2$ be pomsets and suppose that each $p_i \not= 1, q_j \not= 1$. Then $p_1 \parallel p_2 \not= q_1 q_2$ holds. 
\end{lem}

\begin{proof}
Suppose that $p_1 \parallel p_2 = q_1 q_2$ holds, and let $(Z, \le)$ be a poset defining $q_1 q_2$.  Thus $Z$ can be partitioned non-trivially as
$Z= V_1 \uplus V_2 = W_1 \uplus U_2$, where $x_1 \le x_2$ if each $ x_i \in V_i$ and $y_1, y_2$ are incomparable with respect to $\le$  if each $ y_i \in W_i$. Suppose $ W_1 \subseteq V_2$; then $ W_2 \supseteq V_1$, giving a contradiction since the sets $V_i, W_i$ are non-empty and so $ W_1 \cap V_2, \, W_2 \cap V_1 \not= \emptyset$. Thus $ W_1 \nsubseteq V_2$ and so $ W_1 \cap V_1 \not= \emptyset$. Similarly $ W_2 \cap V_2 \not= \emptyset$ also holds, again giving a contradiction. Thus the conclusion follows. \qed
\end{proof}

 \begin{lem}[uniqueness of pomset decomposition] \label{lem.pomset.decomp.unique} \mbox{}
 \begin{enumerate}
\item 
Let $p_1 \parallel \cdots \cdots \parallel p_m = q_1 \parallel \cdots \cdots \parallel q_n $ be a pomset and assume that no pomset $p_i$ or $q_j$ is parallel. Then $m=n $ and there is a permutation $ \theta$ on $ \{ 1, \ldots ,m \}$ such that each $p_i = q_{ \theta(i)}$. 
\item
Let $p_1  \ldots \ldots  p_m = q_1  \ldots \ldots q_n $ be a pomset and assume that no pomset $p_i$ or $q_j$ is sequential. Then $m=n $ and each $ p_i = q_i$. 
\end{enumerate}
 \end{lem}

 \begin{proof} 
(2) is proved in Gischer~\cite[Lemma 3.2]{Gischer:1988:EqThPo}.  (1) is proved as follows. Let $(V, \le)$ be a poset defining $p_1 \parallel \ldots \ldots \parallel p_m$. We may assume that $V \not= \emptyset$ since otherwise the conclusion is obvious. 
 We may define the partition $V= V_1 \uplus \ldots \uplus V_m$, where each pomset $p_i$ is defined by $V_i  \not= \emptyset $ and the restriction of $\le$ to $V_i$.
 Similarly, $V= W_1 \uplus \ldots \uplus W_n$, where each pomset $q_i$ is defined by $W_i \not= \emptyset  $ and the restriction of $\le$ to $W_i$.
Define the collection 
$$ S = \{ X \subseteq V \vert \; x \in X \wedge y \in V - X \Rightarrow \neg(x \le y \vee y \le x)  \}.$$ 
Clearly $X,Y \in S \Rightarrow X \cap Y \in S$ holds. 
Owing to the indecomposability conditions on $p_i$ and $q_j$, the sets $V_i, W_j$ are  minimal non-empty elements of $S$ and so $V_i \cap W_j \not= \emptyset \Rightarrow  V_i = W_j  $ holds, proving the result. \qed
 \end{proof}
 
 Corollary \ref{homo.4980r.plant.bimonoid.cor} states that  the pomset language defined by $T_{bimonoid}( \Sigma)$  is the free bimonoid over $ \Sigma$.
 
 \begin{cor}  \label{homo.4980r.plant.bimonoid.cor}
 Let  $ \Sigma $ be an  alphabet, let $M$ be a bimonoid and let $\kappa: T_{bimonoid}(\Sigma) \to M$ be a  homomorphism of the bimonoid operations. Let $t,t' \in   T_{bimonoid}(\Sigma)$  with $ \llbracket t \rrbracket =  \llbracket t' \rrbracket$. Then $ \kappa(t) = \kappa(t')$ holds. 
 \end{cor}

\begin{proof}
Using Theorem \ref{flamingo.5925g.lem} and Lemma \ref{lem.pomset.decomp.unique} it follows by induction on the structure of $t$  that $t=t'$ holds in any bimonoid, and hence in $M$. \qed
\end{proof}

Our main result of the subsection follows.

\begin{lem} \label{236w.shrew.lem}
Let  $ \Sigma$ be an  alphabet, let $M$ be a bimonoid and let $\kappa: T_{bi-KA}(\Sigma) \to 2^M$ be a homomorphism of the bi-Kleene operations.  Suppose we extend $ \kappa$ to $\pom_{sp}( \Sigma)$  by defining $ \kappa(p)= \kappa(t)$ for any  $t \in T_{bimonoid}( \Sigma)$ with $ \llbracket t \rrbracket = \{p \}$ (well-defined by Corollary \ref{homo.4980r.plant.bimonoid.cor}). 
Let $t \in T_{bi-KA}( \Sigma)$. Then 
  $$ \llbracket \kappa(t) \rrbracket = \bigcup_{  p \in  \llbracket t \rrbracket    } \llbracket \kappa(p) \rrbracket $$ holds. In particular, $ \llbracket t \rrbracket = \llbracket t' \rrbracket \Rightarrow  \kappa(t) = \kappa(t')$ holds, and hence $\kappa$ defines a bi-Kleene homomorphism from $ \big\{  \lb t \rb \big\vert \, t \in T_{bi-KA}(\Sigma) \big\}$ into $2^M$.  
\end{lem}

\begin{proof}
The displayed equation follows by induction on the structure of $t$. If $t \in \Sigma \cup \{0,1 \} $ then the equality is obvious, and the case where $t= t_1 + t_2$ is straightforward. We now consider the remaining cases. 
\begin{itemize}
\item 
Suppose that $t= t_1 t_2$. Then 
\begin{align*}
  \llbracket \kappa(t) \rrbracket   =
   \llbracket \kappa(t_1t_2) \rrbracket = \llbracket \kappa(t_1) \kappa( t_2) \rrbracket = \llbracket \kappa(t_1) \rrbracket \llbracket \kappa(t_2) \rrbracket & = \\
      (\bigcup_{ p_1 \in  \llbracket t_1 \rrbracket    } \llbracket \kappa(p_1) \rrbracket ) \;\;  (\bigcup_{ p_2 \in \llbracket t_2 \rrbracket  } \llbracket \kappa(p_2) \rrbracket ) =
   \bigcup_{p_1 \in  \llbracket t_1 \rrbracket, \, p_2 \in  \llbracket t_2 \rrbracket  } \llbracket \kappa(p_1p_2) \rrbracket &=
    \bigcup_{ p \in  \llbracket t \rrbracket  } \llbracket \kappa(p) \rrbracket 
 \end{align*}  follows, using the inductive hypothesis for each $t_i$ at the fourth equality. 
\item
Suppose that $t= s^* $. Then    \begin{align*}  \llbracket \kappa(t) \rrbracket &= 
\bigcup_{n \ge 0}            \llbracket \kappa(s) \rrbracket^n  
 \\ &=   
   \bigcup_{n \ge 0}  \Big(   \big( \bigcup_{p_1 \in  \llbracket s \rrbracket   } \llbracket \kappa(p_1) \rrbracket \big)            \ldots  \big( \bigcup_{p_n \in  \llbracket s \rrbracket   } \llbracket \kappa(p_n) \rrbracket \big)                         \Big)   
  \\ &=
    \bigcup_{n \ge 0} \;\; \bigcup_{ \text{ each } p_i \in  \llbracket s \rrbracket   } \llbracket \kappa(p_1) \rrbracket \ldots  \llbracket \kappa( p_n) \rrbracket =  
 \bigcup_{n \ge 0} \;\; \bigcup_{q \in  \llbracket s \rrbracket^n   } \llbracket \kappa(q) \rrbracket 
 \\& = 
  \bigcup_{ q \in  \llbracket s^* \rrbracket  } \llbracket \kappa(q) \rrbracket,  
    \end{align*} using the inductive hypothesis at the second equality. 
\end{itemize}
The cases where $t= t_1 \parallel t_2$ or $ t= s^{(*)}$ are similar to those above, hence the conclusion holds. \qed
\end{proof}

Lemma  \ref{236w.shrew.lem} has analogues for $ T_{Reg}(\Sigma) $ and monoids, and  $ T_{ComReg}(\Sigma) $ and commutative monoids, and these have similar proofs.

\subsection{Depth of  a series-parallel  pomset}
 
 In order to prove our main theorems, we need to find a quasi-partial order on bi-Kleene terms in such a way that a parallel term is preceded by its sequential subterms and ground subterms (and the analogous statement with sequential and parallel interchanged also holds)  and this ordering 
is determined by the language that a term defines. Therefore, we first define the \textit{depth} of a pomset, and then 
 extend this definition to bi-Kleene terms.

 \begin{defn}[depth of a series-parallel pomset] \rm
 Let $p \in \pom_{sp}$. Then we define $  \depth(p) \in \Bbb{N}$ recursively as follows.
 \begin{itemize}
 \item
 If $p$ is a singleton pomset or $p= 1 $, then $  \depth(p) =0$.
 \item
 If $ p= p_1 \parallel \ldots \ldots \parallel p_m$ for $ m \ge 2$ and each $p_i$ is a singleton pomset or  is sequential, then
 $$  \depth(p) = max_{i \le m} \; \depth(p_i) + 1.$$ 
 \item
 If $ p= q_1 \ldots \ldots q_n$  for $ n \ge 2$  and each $q_i$  is a singleton pomset or  is sequential, then $$  \depth(p) = max_{i \le n} \;  \depth(q_i) + 1.$$
 \end{itemize}
 Owing to Lemma \ref{lem.pomset.decomp.unique} and Lemma \ref{flamingo.5925g.lem}, this is a valid definition. 
 \end{defn}

\begin{defn}[width of a pomset] \rm
The width of a pomset $p$, $ \width(p)$, is the maximal cardinality of any set of wholly unordered vertices in a representation of $p$. If $L$ is a pomset language then $ \width(L)$ is the maximum width of any pomset in $L$, if this is defined, in which case we say that $L$ has bounded width; otherwise we define $\width(L) = \infty$.  We also define $\width(t) = \width (\lb t \rb)$ for a bi-Kleene term $t$. 
\end{defn}

Observe that if $t \in T_{bi-KA}(\Sigma)$ and  $ \lb t \rb$ has bounded width, then   $ \lb t \rb = \lb t' \rb$ for some $t' \in T_{bw-Rat}(\Sigma)$,  since any subterm $s^{(*)}$ of $t$ can be replaced by 
the term $ \sum_{i=0}^{\width(t)} s^{(i)}$, thus eliminating occurences of $^{(*)}$ from $t$. Conversely, every term in $ T_{bw-Rat}(\Sigma)$ defines a language of bounded width. This justifies our bw-rational terminology.

\subsection{Standardising terms using the bi-Kleene axioms}

In this subsection we will show that the parallel and sequential subsets of a rational language are rational, and definable by terms that can be computed. There is a difficulty, however, with the usual Kleene operations in that the way to partition a rational language into its parallel, sequential and other pomsets is not clearly indicated by the highest-level operation that defines it; for example, a language $ \lb t^* \rb$ may contain both  parallel and  sequential pomsets. Therefore we consider new unary  operations 
 $^!$, $^{(!)}$  that will not be used outside this subsection. They are defined by 
\begin{equation}
 u^!= u^* u^2,  \qquad  u^{(!)}= u^{(*)} \parallel u^{(2)}. \label{kleene_star.recover.eqn}
\end{equation}  

  Definition  \ref{defn.term.equiv.relations} gives the relations between terms with which our main theorems will be expressed. 

\begin{defn}[The $=_{bi-KA}$ and   $ =_{bw-Rat} $relations] \label{defn.term.equiv.relations} \rm
Let 
$\Sigma$ be an alphabet and let $t,t' \in T_{bi-KA}(\Sigma)$. We say that $ t =_{bi-KA}  t'$ if $t=t'$ holds in every bi-Kleene algebra. If $t,t' \in T_{bw-Rat}(\Sigma)$, then we say 
 $ t =_{bw-Rat}  t'$ if $t=t'$ holds in every bw-rational algebra. 
  We also  define the partial orderings $ \le_{bi-KA}$ and  $ \le_{bw-Rat}$ by analogy with 
  (\ref{prec.abbrev.eqn}). 
\end{defn}

Proposition \ref{one.absent.seq.para.prop} shows the use of defining the new operations given in (\ref{kleene_star.recover.eqn}). 

\begin{prop} \label{one.absent.seq.para.prop}
Let $ \Sigma$ be an alphabet and let $t$ be a term over $\Sigma$ with operations in $ \{ +, \cdot, ^!, \parallel,  ^{(!)} \}$. We extend the definition of the language  $ \lb t \rb$ by interpreting $  ^! ,  ^{(!)} $ as given in (\ref{kleene_star.recover.eqn}). Then $ 1 \notin \lb t \rb$; also, if the term 
 $t=uv$ or $ t= u^!$, then  $ \lb t \rb$ is a sequential language, and an analogous assertion holds for the operations $ \parallel, \, ^{(!)}$. 
\end{prop}

\begin{proof}
The proof that $ 1 \notin \lb t \rb$  follows by induction on the structure of $t$; in particular, it follows from (\ref{kleene_star.recover.eqn}) that  $1 \notin \lb r \rb \Rightarrow 1 \notin \lb r^! \rb,$ and analogously  for $r^{(!)}$, if $r$ has operations in $ \{ +, \cdot, \parallel, ^!, ^{(!)} \}$. The remaining assertions follow by applying this result to $u$ and $v$. \qed
\end{proof}


\begin{prop} \label{transform.term.prop}
Let $ \Sigma$ be an alphabet and let $t \in T_{bi-KA}( \Sigma)$. Suppose the relation $=_{bi-KA}$ is extended to terms containing the unary operations $^!,  ^{(!)}$ by  assuming the substitutions indicated by (\ref{kleene_star.recover.eqn}). Then there  is a term $t' $ with operations in   $ \{0, 1,+, \cdot, \parallel, ^!, ^{(!)} \}$ satisfying $t =_{bi-KA} t' $ such that either $t'=0$ or $0$ does not occur in $t'$ and $1$ does not occur in the argument of any operation except possibly $+$ in $t'$. 
\end{prop}

\begin{proof}
By using the Kleene-valid substitutions
\begin{equation}
u+0 = 0+u \to u, \qquad u0=0u \to 0, \qquad 0^* \to 1, \label{zero.elim.eqn}
\end{equation}
and their parallel analogues, we may assume that either $t=0$ or  $0$ does not occur in $t$. 
We now eliminate the iteration operations $ ^*, ^{(*)}$ from $t$ by replacing them with new unary operations $^!$, $^{(!)}$ respectively using the following identities;  
\begin{equation}
     u^*=  u^! + 1 + u , \qquad u^{(*)}=  u^{(!)} + 1 + u,  \label{kleene_star.elim.eqn}
\end{equation}  
which follow from (\ref{kleene_star.recover.eqn}) plus the Kleene axioms.
If $t \not= 0$, then by using the distributive laws and the  substitutions
\begin{equation}
u1= 1u \to u, \qquad (u+1)^! = (1+ u)^! \to u^! + 1 +u,  \label{one.elim.eqn}
\end{equation}
which follow from the Kleene axioms plus (\ref{kleene_star.recover.eqn}), and their parallel analogues,  we can ensure that   $1$ does not occur in the resulting term in the argument of any operation except possibly $+$, thus proving the result. \qed 
\end{proof}

We are now able to show that a rational language can be expressed as a sum of terms representing its sequential, parallel and remaining pomsets. 

\begin{lem} \label{rat.lang.jay249y.partition.lem}
Let $\Sigma$ be an alphabet and let $ t \in T_{bi-KA}(\Sigma)$. Then  the pomset languages
   $  \lb t \rb \cap \para_i$ for each $ i \ge 1$ are rational and definable by terms that are computable from $t$;  and there exist terms $t', t'',t'''  \in T_{bi-KA}(\Sigma)$ that that are computable from $t$  and define pomset languages
    $ \lb t \rb \cap \seq$, $  \lb t \rb \cap \para   $ and $\lb t \rb \cap (\Sigma \cup \{1 \})$
     and satisfy 
$$  t =_{bi-KA}  t' + t'' + t'''.$$
Furthermore, $ \depth(t) < \infty$.
\end{lem}

\begin{proof}
By Proposition \ref{transform.term.prop},  we may assume that either $t=0$ or  $1$ does not occur in $t$ in the argument of any operation except possibly $+$, and $t$ has operations lying in $ \{ 1,+, \cdot, \parallel, ^!, ^{(!)} \}$. 
 We prove the results (apart from the computability assertions, which follow immediately) for  the set of terms $t$ satisfying these conditions  by induction on the structure of $t$, and we can then reinstate the operations $^*, ^{(*)}$ in $t',t'',t'''$  using (\ref{kleene_star.recover.eqn}).

If $t \in \Sigma \cup \{0,  1 \}$ then the results are immediate. 
If $t= t_1 + t_2$ then the results follow from the inductive hypothesis applied to each $t_j$.

 If  $t=u \parallel v$ then by Proposition \ref{one.absent.seq.para.prop} the term $t$ is parallel and so 
$ L = L \cap \para$ holds; also,
 $$ \lb u \parallel v \rb \cap \para_i =  \bigcup_{j+k=i}  \lb u  \rb \cap \para_j \parallel    \lb v  \rb \cap \para_k $$
 and                      
 $$ depth(u \parallel v) \le \depth(u) +  \depth(v) +1,$$
        proving the  rationality assertion for the languages  $  \lb t \rb \cap \para_i$ and the depth assertion for $t$ by the inductive hypothesis.   The case   $t=u  v$ is analogous. 
 
 If instead  $t= u^{(!)}$  then again by Proposition \ref{one.absent.seq.para.prop}, $t$ is  parallel and so 
$ L = L \cap \para$ holds; also,
 $$ \lb u^{(!)} \rb \cap \para_i =
  \bigcup_{j \le i}  \,\,  \bigcup_{2 \le k_1 + \cdots + k_j = i} \lb u  \rb \cap \para_{k_1} \parallel \cdots \parallel  \lb u  \rb \cap \para_{k_j}$$
  and
   $$ \depth(u^{(!)}) \le \depth(u)+1$$
        proving the  rationality assertion for the languages  $  \lb t \rb \cap \para_i$ and                      
  the depth assertion for $t$ by the inductive hypothesis.  The case  $t= u^{!}$ is analogous. 
\qed
\end{proof}

Proposition \ref{devil.386f.kite.prop} will be an essential tool for proving assertions on bi-Kleene terms by induction on the depth of their languages. 

\begin{prop} \label{devil.386f.kite.prop}
Let $\Sigma$ be an alphabet and let $ t \in T_{bi-KA}(\Sigma)$. If 
$t$ is parallel, then $ t =_{bi-KA} c(u_1, \ldots, u_m)    $ for a commutative-regular term $c$ and terms $u_i \in T_{bi-KA}(\Sigma)$ defining non-empty languages that are either sequential or lie in $\Sigma$, and satisfy \\
 $ \depth(u_i) <  \depth(t)$, with $c$ and each $u_i$ being computable from $t$.  
\end{prop}

\begin{proof} 
By Proposition \ref{transform.term.prop}, we may assume that  
  either  $t=0$, or $t$ contains only the operations $1, +, \cdot,  ^!, \parallel, ^{(!)}$, with $1$ not occurring in the argument   of any operation in $t$ except possibly $+$. If $t=0$ then the conclusion is obvious, so we assume the latter case.  Since $t$ is parallel, this implies that $1$ does not occur at all in $t$.
Thus  $  t $ has the form  $   c(u_1, \ldots, u_m)    $ for a  term $c$ with operations in $\{ +, \parallel, ^{(!)} \}$ and terms $u_i$ that either lie in $ \Sigma$ or have the form $uv$ or $u^!$ and are hence  sequential by  Proposition \ref{one.absent.seq.para.prop}, and define non-empty languages.  For each $ j \le m$, let $p_j$ be a pomset in  $ \lb u_j \rb$ of maximal depth.  We may assume that each $u_i$ actually occurs in $t$. 
Let $ i \le m$. Thus for an alphabet $ \{ \sigma_1, \ldots, \sigma_m\}$, the language  $ \lb  c(\sigma_1, \ldots, \sigma_m)  \rb  $ contains a parallel word $w$ of width $ \ge 2$ in which $ \sigma_i$ occurs, and so the pomset language  
$ \lb t \rb$ contains  $w( \sigma_j \setminus p_j \vert \, j \le m) $, whose depth is greater than that of $u_i$,  proving the depth assertion. 
 By reinstating $^*$ and $^{(*)}$ in each $u_i$ and $^{(*)}$ in $c$ using (\ref{kleene_star.recover.eqn}), we get the result required.
\qed
\end{proof}

\subsection{Regular and commutative-regular languages are closed under boolean operations} 

\label{subsect.bool.(com)reg.closed}
 
Theorem \ref{closed.under.comp.729u.thm} recalls the fact that our first main theorem is known to hold for the subclasses of regular and commutative-regular languages.  

\begin{thm} \label{closed.under.comp.729u.thm}
Let $ \Sigma$ be an alphabet and let $t_1,t_2 \in T_{Reg}( \Sigma) $, or alternatively $t_1,t_2 \in T_{ComReg}( \Sigma) $. Then there exists a  term $s \in  T_{Reg}( \Sigma)$ or $ T_{ComReg} $, respectively, such that $ \llbracket s \rrbracket = \llbracket t_1 \rrbracket - \llbracket t_2 \rrbracket$. Furthermore, $s$ can be computed from $t_1$ and $t_2$.
\end{thm}

\begin{proof}
If each term $t_i$ is regular, then the conclusion is a well-known theorem for regular languages. If each term $t_i$ is commutative-regular, then it
 follows from Conway \cite[Chapter 11]{ConwayJH:71:regafm}, the computability result being an implicit consequence of his method of proof. \qed
\end{proof}

\begin{cor} \label{2357y.triton.cor}
Let $ \Sigma$ be an alphabet and let $t_1,t_2 $ be both regular or both commutative-regular terms over $ \Sigma$. Then it is decidable whether $\llbracket t_1 \rrbracket = \llbracket t_2 \rrbracket$ holds. 
\end{cor}

\begin{proof}
This follows since 
$$\llbracket t_1 \rrbracket = \llbracket t_2 \rrbracket \iff  (\llbracket t_1 \rrbracket - \llbracket t_2 \rrbracket) \cup (\llbracket t_2 \rrbracket - \llbracket t_1 \rrbracket) = \emptyset  $$
holds, and it is clearly possible to decide whether an element of $T_{bi-KA}(\Sigma)$ defines the empty language.  \qed
\end{proof}

\begin{cor} \label{pure.term.138y.basis.cor} 
Let  $ \Sigma$ be an alphabet and 
let $T$ be a finite set of  elements of  $T_{bi-KA}(\Sigma)$ that are either all regular  or all commutative-regular.  Then there exists a finite set $U$ of    terms, pairs of which define  disjoint languages, and  such that for each $t \in T$, there exists  $ V_t \subseteq U$ such that $ \llbracket t \rrbracket = \bigcup_{x \in V_t} \llbracket x \rrbracket$ holds. Furthermore, the set $U$ can be computed from $T$, as can the subset $ V_t$ from $T$ and $t$. 
\end{cor}

\begin{proof} 
Write $T = \{ t_1, \ldots, t_n\}$. By Theorem \ref{closed.under.comp.729u.thm}, for each  $ N \subseteq \{ 1, \ldots, n\}$, we may define a term $s_N$  satisfying 
$\llbracket s_N \rrbracket =  \bigcup_{i \in N} \llbracket t_i \rrbracket - \bigcup_{i \notin N} \llbracket t_i \rrbracket$, and $M \not= N \Rightarrow \llbracket s_M \rrbracket \cap   \llbracket s_N \rrbracket = \emptyset$ holds. Clearly $$\llbracket t_i \rrbracket = \bigcup_{i \in N }   \llbracket s_N \rrbracket,$$
  thus proving the Corollary, since from  Theorem \ref{closed.under.comp.729u.thm}, the terms $s_N$ can clearly be computed from $T$.
\qed
\end{proof}

\section{Closure of rational pomset languages  under Boolean operations}

\label{sect.bool.closed.pomset.sp}

In this section we prove our first main theorem.

\subsection{The label set $ L_U$ and function   $ \nu $ }

\label{195l.cougar.theta.section}

  For the remainder of this section, and in Section \ref{sect.biKleene.terms.equal.imply.langs}, Definition \ref{777.kanga.defn} will be assumed. 

\begin{defn}[associating a label with a term, the function $\nu$] \label{777.kanga.defn} \rm 
For any term $u$, we assume a label $l_u$, where distinct terms define distinct labels, and for any set $U$ of terms over an alphabet $ \Sigma$, we define $ L_U=  \{ l_u \vert \, u \in U \}$. We also define the homomorphism
$$ \nu : T_{bi-KA}(L_U) \to  T_{bi-KA}( \Sigma)$$ 
given by $ \nu(l_u) = u$. Further, for any $p \in \pom_{sp}(L_U)$, we define $ \nu(p) = \nu(t)$, where  $t \in T_{bimonoid}(L_U)$ satisfies $ \llbracket t \rrbracket = \{p \} $ (well-defined by Corollary \ref{homo.4980r.plant.bimonoid.cor}). 
\end{defn}

\textbf{Note:} the assertions of Proposition \ref{devil.386f.kite.prop},  Lemma  \ref{574a.mouse.lem} and Corollary \ref{331f.eagle.cor}, and Lemma \ref{bool.closure.225m.pig.lem} in Section \ref{sect.biKleene.terms.equal.imply.langs} have their counterparts with references to sequential and parallel multiplication interchanged, and these have analogous proofs.

\begin{lem} \label{574a.mouse.lem}
Let  $ \Sigma$ be an alphabet and let $U$ be a set of  elements of $ T_{bi-KA}( \Sigma)$ such that every element of $U$ either lies in $ \Sigma$ or is sequential.  Assume that distinct terms in $U$ define disjoint languages. Let  
$ p $  be a parallel product of elements of $L_U$ 
 and let $s  \in T_{ComReg}( L_U )$. Then  
 $$ p \notin \llbracket s \rrbracket  \Rightarrow \llbracket \nu(p) \rrbracket \cap \llbracket \nu(s) \rrbracket = \emptyset $$ holds. 
\end{lem}

\begin{proof} 
Order the terms $s$, firstly by the total number of occurrences of $+$ and $^{(*)}$, and secondly by the number of occurrences of  $ \parallel$. Assume that $ p \notin \llbracket s \rrbracket$ holds.  We prove $ \llbracket \nu(p) \rrbracket \cap \llbracket \nu(s) \rrbracket = \emptyset$  by induction using this ordering. 
  \begin{itemize}
  \item
  Suppose that $s \in L_U$. If the commutative word $p \notin L_U \cup \{1\} $, then we may write $p= q \parallel q'$ for $q,q' \not= 1$ and hence $\nu(p)=  \nu(q) \parallel \nu(q')$ is a parallel term, whereas no elements of  $\llbracket \nu(s) \rrbracket$    are parallel, proving $ \llbracket \nu(p) \rrbracket \cap \llbracket \nu(s) \rrbracket = \emptyset $. On the other hand, if  $p \in L_U$ or $ p =1$, then   $ \llbracket \nu(p) \rrbracket \cap \llbracket \nu(s) \rrbracket = \emptyset $ follows, respectively, from the disjointness assumption on the elements of $U$ or the fact that $ 1 \notin  \llbracket u \rrbracket$ for all $ u \in U$.  
  \item
   Suppose $s = s_1 + s_2$. The conclusion follows by the inductive hypothesis applied to each term $s_i$. 
\item 
Suppose  $s= s_1 \parallel s_2$.   Write $ p= l_{u_1} \parallel \ldots \parallel l_{u_m}$ with each $u_i \in U$.
Assume the conclusion is false for $s$; thus  there are pomsets $q_i \in \llbracket \nu(s_i) \rrbracket$
such that 
 $ q_1 \parallel q_2  \in \llbracket \nu(p) \rrbracket \cap \llbracket \nu(s) \rrbracket$. 
  Since every element of every set $\llbracket u_i \rrbracket$ is not parallel and not $1$, after rearrangement of the labels $l_{u_i}$  we may write $ q_1 = v_1 \parallel \ldots \parallel v_n$ and $ q_2 = v_{n+1} \parallel \ldots \parallel v_m$ with each pomset $v_i \in 
  \llbracket u_i \rrbracket$. 
Thus  $ l_{u_1} \parallel \ldots \parallel l_{u_n} \in \llbracket s_1 \rrbracket$ and  $ l_{u_{n+1}} \parallel \ldots \parallel l_{u_m} \in \llbracket s_2 \rrbracket$ by the inductive hypothesis, and so $p \in \llbracket s_1 \parallel s_2 \rrbracket$, giving a contradiction. 
 \item 
  Suppose $s= r^{(*)}$. Thus for every $n \ge 0$,  $\llbracket p \rrbracket \cap \llbracket \underbrace{r \parallel \ldots \parallel r}_{ n \text{ terms}}  \rrbracket = \emptyset    $ holds, and  
    from the minimality condition on $s$, $ \llbracket \nu(p) \rrbracket \cap \llbracket \nu( \underbrace{r \parallel \ldots \parallel r}_{ n \text{ terms}} ) \rrbracket = \emptyset $ follows. 
     Since $ \llbracket \nu(s) \rrbracket = \cup_{n \ge 0} \llbracket \nu(  \underbrace{r \parallel \ldots \parallel r}_{ n \text{ terms}} ) \rrbracket$, this leads to a contradiction. \qed
  \end{itemize} 
\end{proof}

Corollary \ref{331f.eagle.cor} extends 
Lemma \ref{574a.mouse.lem} by  replacing $p$ by an arbitrary  commutative-regular term.

\begin{cor} \label{331f.eagle.cor}
Let  $ \Sigma$ be an alphabet and let $U$ be a set of  elements of $ T_{bi-KA}( \Sigma)$ such that every element of $U$ either lies in $ \Sigma$ or is sequential.  Assume that distinct terms in $U$ define disjoint languages.  
 Let  $s,s'  \in T_{ComReg}( L_U )$. Then 
 $$ \llbracket s \rrbracket  \cap \llbracket s'  \rrbracket = \emptyset \, \Rightarrow \, \llbracket\nu(s)  \rrbracket \cap  \llbracket \nu(s')   \rrbracket = \emptyset $$ 
 holds.
\end{cor}

\begin{proof}
 If $\llbracket \nu(s) \rrbracket$ and $ \llbracket \nu(s') \rrbracket$ are not disjoint, then from the commutative-regular analogue of Lemma \ref{236w.shrew.lem}, there exists a commutative word $w$ such that  $  w\in \llbracket s \rrbracket  $ and  $\llbracket \nu(w) \rrbracket \subseteq  \llbracket \nu(s) \rrbracket$ and  $\llbracket \nu(w) \rrbracket$ intersects with $\llbracket \nu(s') \rrbracket$, and so from Lemma \ref{574a.mouse.lem},  $  w\in \llbracket s' \rrbracket  $ also follows.
 \qed 
\end{proof}

\begin{lem}  \label{fly.266g.lem}
Let  $ \Sigma$ be an alphabet and let $T$ be a finite set of  elements of $ T_{bi-KA}( \Sigma)$.  Then there exists a finite set $U$ of elements of $ T_{bi-KA}( \Sigma)$ defining non-empty pairwise disjoint languages,   such that for each $t \in T$, there exists  $ U_t \subseteq U$ such that $ \llbracket t \rrbracket = \bigcup_{x \in U_t} \llbracket x \rrbracket$ holds.  Furthermore, the set $U$ can be computed from $T$ and any subset $U_t$ can be computed from $T$ and $t$. 
\end{lem}

\begin{proof}
We will prove the computability assertion separately; first we prove the preceding claims in the Lemma by induction on 
   $  \depth( \sum_{x \in T} x  )$. If $ T \subseteq  \Sigma \cup \{  1\} $ then the conclusion is obvious,  and so 
using Lemma \ref{rat.lang.jay249y.partition.lem}, we need only consider the case that each term in $T$ is parallel; the case   that each term in $T$ is sequential is analogous.

  By Proposition  \ref{devil.386f.kite.prop}, for each $t \in T$ there exists a  finite set $U_t$ of terms  that all either lie in $\Sigma$ or are sequential and a commutative-regular term $s_t$ over $ L_{U_t}$ such that $ \llbracket    t \rrbracket  = \llbracket \nu(s_t)\rrbracket   $ and for each $u \in U_t$,  $ \depth(u) <   \depth(t)$, and hence  
  \begin{equation}
    \depth(\sum_{ x \in \cup_{t \in T} U_t} x) \, <  \,  \depth( \, \sum_{x \in T} x )  \label{alpha.eqn} 
\end{equation}  holds. 
  
  From applying the inductive hypothesis to   $ \cup_{t \in T} U_t$  there is a set $V$ of  terms over $ \Sigma$ defining non-empty pairwise disjoint pomset languages,  and  such that for each $u \in \cup_{t \in T} U_t$, there exists  $ V_u \subseteq V$ such that $$ \llbracket u \rrbracket = \bigcup_{x \in V_u} \llbracket x \rrbracket$$ holds.
  
   For each $t \in T$, let $s_t'$ be obtained from $s_t$ by replacing every letter $l_u$ by the sum 
 $ \sum_{ x \in V_u} l_x$. Thus  $ \llbracket \nu(s_t')\rrbracket = \llbracket \nu(s_t)\rrbracket   = \llbracket t \rrbracket$ holds by Theorem \ref{kozen.390o.freeness.reg.comreg.thm}.  By Corollary \ref{pure.term.138y.basis.cor} applied to the terms $s_t'$, there is a set $C$ of  commutative-regular terms defining non-empty pairwise disjoint  languages and such that for each $t \in T$, there are terms $ c_1, \ldots, \ldots, c_n \in C$ satisfying $  \llbracket s_t'\rrbracket = \llbracket c_1 + \ldots \ldots + c_n \rrbracket$
 and again from  Theorem \ref{kozen.390o.freeness.reg.comreg.thm},  
 $$ \llbracket t \rrbracket =  \llbracket \nu(s_t')\rrbracket = \llbracket \nu(c_1) + \ldots \ldots + \nu(c_n) \rrbracket$$
 holds.
  From  Corollary \ref{331f.eagle.cor}, the terms in $ \nu(C)$ also satisfy the required disjointness property and hence satisfy the conclusion of the Lemma for $U$.

We now consider the computability assertion. We define a recursive algorithm $ \mathcal{A}$ that on input $T$ computes the sets $U$ and $U_t$ for each $ t \in T$ satisfying the conditions required. 
We may assume that each term in $T$ defines a non-empty pomset language. 
$ \mathcal{A}$ is defined precisely as indicated by our proof above. We prove by induction on $   \depth( \sum_{x \in T} x) $ that  $ \mathcal{A}$ terminates with the correct outputs. We define the partition $T=  T_{para} \uplus T_{seq} \uplus T_{  \Sigma  } $, where $T_{para}$ contains all elements of $T$ that  are parallel, $T_{seq}$ contains all elements of $T$ that are sequential, and 
$ T_{  \Sigma  } \subseteq \Sigma \cup \{1 \} $   contains all remaining elements of $T$.   The term sets $U_t$ and terms $s_t$ can  be computed from each  $t \in T_{para}$, by  Lemma \ref{devil.386f.kite.prop}.  $ \mathcal{A}$ obtains the sets $V$ and $V_u$ for each $u \in \cup_{t \in T_{para}} U_t$ by calling itself with input $\cup_{t \in T_{para}} U_t$;  by the inductive hypothesis and (\ref{alpha.eqn}),  $ \mathcal{A}$ terminates and returns the correct values. 
 Thus the terms $ s_t'$ can also be computed, and so the set $C$ and the appropriate set of elements $\{c_1, \ldots, c_n\}$ for each term $t$   can be computed by Corollary \ref{pure.term.138y.basis.cor}.  The function $ \nu $ is clearly computable and thus $ \mathcal{A}$ returns the correct term sets for $T_{para}$.  The correct output for $T_{seq}$ is computed analogously. \qed
\end{proof}

Our first main Theorem now follows.

\begin{thm} \label{bool.closure.225m.thm}
Let  $ \Sigma$ be an alphabet and let $t_1,t_2 \in T_{bi-KA}( \Sigma ) $. Then there exist elements of $ T_{bi-KA}( \Sigma)$ defining the sets $ \llbracket t_1\rrbracket \cup  \llbracket t_2\rrbracket$, $ \llbracket t_1\rrbracket \cap  \llbracket t_2\rrbracket$ and $\llbracket t_1\rrbracket -  \llbracket t_2\rrbracket$, which can be computed from $t_1$ and $t_2$. 
\end{thm}

\begin{proof}
The case $ \llbracket t_1\rrbracket \cup  \llbracket t_2\rrbracket$ is trivial, and since $ \llbracket t_1\rrbracket \cap  \llbracket t_2\rrbracket  =  \llbracket t_1 +t_2 \rrbracket - (\llbracket t_1\rrbracket -  \llbracket t_2\rrbracket) - (\llbracket t_2\rrbracket -  \llbracket t_1\rrbracket)$ holds, 
it suffices to prove the existence of an element $ s \in  T_{bi-KA}( \Sigma)$ such that $\llbracket s\rrbracket = \llbracket t_1\rrbracket -  \llbracket t_2\rrbracket$ holds. This follows from Lemma \ref{fly.266g.lem}  with $T= \{t_1, t_2 \}$ in that Lemma.  \qed
\end{proof}

We now give our bi-Kleene term decidability result.

\begin{thm} \label{biKA.equiv.decide.3870.willow.thm}
Let $ \Sigma $ be an alphabet and let 
 $t, t' \in  T_{bi-KA}( \Sigma ) $. Then 
 it is decidable whether $ \llbracket t\rrbracket  = \llbracket t'\rrbracket $ holds. 
\end{thm}

\begin{proof}
This  follows from Theorem  \ref{bool.closure.225m.thm}, similarly to the proof of Corollary  \ref{2357y.triton.cor}. 
\end{proof}

\section{Equality between bi-Kleene terms defining pomset languages   is a consequence of the bi-Kleene axioms}

\label{sect.biKleene.terms.equal.imply.langs}

In this section we use Lemma \ref{fly.266g.lem} to  prove our second main theorem. 
 We first show that under stricter hypotheses, the converse implication to that given in Corollary \ref{331f.eagle.cor} holds.

\begin{lem} \label{bool.closure.225m.pig.lem}
Let  $ \Sigma$ be an alphabet and let $U$ be a set of  elements of $  T_{bi-KA}( \Sigma)$ such that every element of $U$ either lies in $ \Sigma$ or is sequential and defines a non-empty language. 
 Assume that pairs of terms in $U$ define disjoint languages. Let  $s,t$ be  commutative-regular terms over $L_U$. Then 
$$  \llbracket \nu(s) \rrbracket = \llbracket \nu( t) \rrbracket \Rightarrow  \llbracket s \rrbracket = \llbracket t \rrbracket    $$
holds.
\end{lem}

\begin{proof}
Suppose that  $\llbracket s \rrbracket \not= \llbracket t \rrbracket$ holds. Then there exists a pomset $p \in \llbracket s \rrbracket    -  \llbracket t \rrbracket  $  ($s,t$ may need to be interchanged). Let $ \tilde{p} \in T_{bimonoid}( \Sigma)$ satisfy $ \llbracket \tilde{p} \rrbracket = \{p\}$.   Thus  $ \llbracket s + \tilde{p}  \rrbracket =  \llbracket s \rrbracket$ and so   
$   \llbracket \nu(s)  \rrbracket  + \llbracket \nu(p) \rrbracket =  \llbracket \nu(s) +\nu(\tilde{p} ) \rrbracket =  \llbracket \nu(s) \rrbracket$ by Lemma \ref{236w.shrew.lem},
whereas  $ \llbracket \nu( p ) \rrbracket  \cap   \llbracket \nu( t) \rrbracket = \emptyset$  by the commutative-regular analogue of  Lemma \ref{236w.shrew.lem}. 
Since each element of $U$ defines a non-empty language,  $ \llbracket \nu( {p} ) \rrbracket \not= \emptyset$ and so $\llbracket \nu(s) \rrbracket \not= \llbracket \nu( t) \rrbracket$ follows. \qed
\end{proof}

Our second main theorem follows. 

\begin{thm} \label{cow.1729.term.equal.imply.syntax.thm} 
Let $ \Sigma $ be an alphabet and let 
 $t, t' \in  T_{bi-KA}( \Sigma ) $.
Assume $ \llbracket t\rrbracket  = \llbracket t'\rrbracket $;  then $t    =_{bi-KA}  t' $ holds.
  
\end{thm}

\begin{proof} 
   We  prove   the Theorem 
    by induction on   $  \depth(t)=   \depth(t')$. 
 If $ \llbracket t \rrbracket =   \llbracket t' \rrbracket \subseteq \{ 1 \} \cup \Sigma $, then  
    $t    =_{bi-KA}  t' $ is obvious. By Lemma \ref{rat.lang.jay249y.partition.lem}, we may assume that $t,t'$ are both parallel; the case that they are both sequential is analogous.

    By Proposition  \ref{devil.386f.kite.prop},    there exists a finite set 
 $U$  of terms that all either lie in $\Sigma$ or are sequential  and define  non-empty languages,
  and  commutative-regular terms $ s,s'$ over $ L_U$ such that 
  \begin{equation} t=_{bi-KA} \nu(s), \qquad  t'=_{bi-KA} \nu(s'). \label{a.eqn} \end{equation}
By Lemma  \ref{fly.266g.lem}, there is a finite subset $V  $ of $   T_{bi-KA}( \Sigma)$, pairs of which define disjoint languages, and  such that for each $ u \in U$ there exists $ V_u \subseteq V$ satisfying 
\begin{equation}
 \llbracket u \rrbracket = \bigcup_{x \in V_u}  \llbracket x \rrbracket. \label{a2.eqn}
\end{equation}   
 For each $ u \in U$, let  $w_u$ be a sum of the labels $ l_x$ for each  $ x \in  V_u$. Hence by  Theorem \ref{kozen.390o.freeness.reg.comreg.thm} and (\ref{a2.eqn}),  
\begin{equation}
  \llbracket \nu(w_u) \rrbracket =  \bigcup_{x \in V_u}  \llbracket x \rrbracket =  \llbracket u \rrbracket   \label{b.eqn} 
\end{equation}   holds. 
Let the terms $r,r'$ be obtained from $ s,s'$ respectively by replacing each occurrence of any $l_u \in L_U $ by  $w_u$. Thus $ \nu(r)$ is obtained from $\nu(s)$ by replacing each subterm  $u \in U $ by $ \nu(w_u)$, and similarly for $ \nu(r')$ and $\nu(s')$.
   By Proposition \ref{devil.386f.kite.prop}  and  (\ref{b.eqn}), for each $u \in U$
    $$  \depth(\nu(w_u))=  \depth(u) < \depth(t)= \depth(t')  $$  follows, and so from the inductive hypothesis, 
$\nu(w_u)  =_{bi-KA} u$  follows from (\ref{b.eqn}).  Since $    =_{bi-KA} $ is preserved by congruence, 
 \begin{equation}
   \nu(r)   =_{bi-KA} \nu(s)  =_{bi-KA} t, \qquad  \nu(r')   =_{bi-KA} \nu(s')   =_{bi-KA} t'  \label{c.eqn} 
\end{equation}
holds using  (\ref{a.eqn}). Since  $ \llbracket t\rrbracket  = \llbracket t'\rrbracket $ holds,  $ \llbracket \nu(r) \rrbracket  = \llbracket \nu(r')\rrbracket $ follows from (\ref{c.eqn}). 
  From Lemma \ref{bool.closure.225m.pig.lem},  $ r  =_{bi-KA} r'$ follows from Theorem \ref{kozen.390o.freeness.reg.comreg.thm}   since the terms $r,r'$ are commutative-regular, and so $ \nu(r)  =_{bi-KA} \nu(r')$  holds since  $    =_{bi-KA} $ is preserved by substitution. Hence
 $ t   =_{bi-KA} t'$ follows from (\ref{c.eqn}), thus concluding the proof. 
  \qed
\end{proof}

Theorem \ref{cow.1729.term.equal.imply.syntax.thm} has an analogue for bw-rational algebras.

\begin{thm}  \label{thm.bwrat.equal.imply.bw-equiv}
Let $ \Sigma $ be an alphabet and let 
 $t, t' \in  T_{bw-Rat}( \Sigma ) $.
Assume $ \llbracket t\rrbracket  = \llbracket t'\rrbracket $;  then $t    =_{bw-Rat}  t' $ holds.
\end{thm}

\begin{proof}
This has a similar proof to that of Theorem \ref{cow.1729.term.equal.imply.syntax.thm}. The proof relies on the fact that the proofs of Theorem \ref{cow.1729.term.equal.imply.syntax.thm} and its contributing lemmas and propositions can be adapted for bw-rational algebras by ignoring the cases in their proofs that consider the parallel iteration operation $^{(*)}$. In the case of Theorem \ref{kozen.390o.freeness.reg.comreg.thm}, the relevant result is that the algebra of commutative-word languages generated by an alphabet $\Sigma$ and the operations $ 0,1,+ , \parallel$ is the free idempotent commutative semiring with basis $ \Sigma$, and this is straightforward to prove. 
 \qed
\end{proof}

\section{The bi-Kleene algebra of pomset ideals}

\label{sect.ideals.newdefns.biKleene}

We now move on to considering pomset ideals. 
We first give a criterion for elements of $\pom$ to lie in $\pom_{sp}$. 

\begin{defn}[N-free pomsets] \rm
A pomset defined by vertex set $V$ with partial order $ \le $  is  N-free if  $V$ does not contain a  4-element subset $ \{v_1, \ldots, v_4 \}$ with $ v_1 \le v_2$ and $ v_3 \le v_2, \; v_3 \le v_4$, and such that   $ \le $ when restricted to $\{v_1, \ldots, v_4 \}$ does not contain any other pairs. 
\end{defn}

\begin{thm} \label{2606y.monk.N-free.thm}
A pomset is series-parallel if and only if it is N-free. 
\end{thm}

\begin{proof}
Gischer~\cite[Theorem 3.1]{Gischer:1988:EqThPo}. \qed
\end{proof}

\begin{defn}[ideals of a pomset] \rm
Let $p$ be a pomset. An ideal of $p$ is a pomset  that may be represented using the same vertex set as $p$, with the same labelling, but whose partial ordering is at least as strict as that for $p$. 
Let $L$ be a language of pomsets. Then $\id(L)$ is the  language of  pomsets that are ideals of pomsets lying  in $L$. We say that $L$ is a  (pomset) ideal if $ \id(L) = L$ holds. 
We also define $\id_{sp}(L)  = \id(L) \cap \pom_{sp}$. If $\id_{sp}(L)= L$ then we say that $L$ is an sp-ideal. 
The functions $\id $ and $ \id_{sp}$ are closure operators on the sets $2^{\pom( \Sigma)}$ and $ 2^{\pom_{sp}( \Sigma)}$ respectively \cite[chap.1]{burris1981course}.
\end{defn}

In order to to study the sp-ideal of a parallel product $ p \parallel q$ of pomsets, we introduce the  function $ \odot$. The use of this operation will be demonstrated by the identity (\ref{odot.eqn}).

\begin{defn}[the $ \odot$ binary function on pomset languages] \rm
Let $p_1,p_2$ be pomsets defined with disjoint vertex sets $V_1,V_2$. Then we 
define the set $ p_2 \odot p_2$ to be the set of all  pomsets  $q  \in \pom_{sp}$ whose vertex set is $ V_1 \cup V_2$ and such that $q$ retains the  vertex labelling and ordering of each $p_i$ within  $V_i$. We extend the domain of $ \odot$ pointwise to pairs of pomset languages.  Clearly $ \odot$ is associative and commutative.

\end{defn}

\begin{lem} \label{oslash.852t.para.induct.lem}
 The following hold for pomset languages $L,L', L_j $ for $j$ in an indexing set $S$.  
 \begin{equation}
  \id ( \cup_{j \in S} L_j ) =   \cup_{j \in S} \id( L_j), \label{id.union.eqn}
 \end{equation}
 \begin{equation}
  \id (   LL') =  \id ( L) \id (  L'), \qquad
 \id( L^*) =( \id ( L))^*, \label{id.seq.multiply}
 \end{equation}
 \begin{equation}
  \id(L \parallel L') =  \id(L \parallel \id( L')). \label{id.para.multiply}
 \end{equation}

  Furthermore, if $L,L', L_j \subseteq \pom_{sp} $ then  the same equalities with $\id $ replaced by $ \id_{sp}$ also hold; in addition, 
    \begin{equation}
    \id_{sp} ( L \parallel L') =  \id_{sp} (  L) \odot  \id_{sp} ( L'). \label{odot.eqn}
    \end{equation} 
\end{lem}

\begin{proof}
(\ref{id.union.eqn}) and its $\id_{sp}$ counterpart  follow immediately from the definition of an ideal.
 (\ref{id.seq.multiply}) for $ \id$ is  straightforward. 
 To prove   $\id_{sp} (   LL') \subseteq  \id_{sp} ( L) \id_{sp} (  L')$, observe that any element of  $\id (   LL')$ has the form  $pp'$ with $ p \in  \id ( L)$, $ p' \in \id(L')$. If in addition $pp' \in \pom_{sp}$, then $pp'$ is N-free   by Theorem \ref{2606y.monk.N-free.thm}, hence $p,p'$ are also N-free and again by this Theorem,  $p, \, p' \in \pom_{sp}$ hold. The other inclusion is obvious, and hence 
  $\id_{sp} (   L^{i}) =  (\id_{sp}(L))^{i}$ for each $i \ge 0$ follows by induction on $i$. Thus 
  $  \id_{sp}( L^*) =( \id_{sp} ( L))^*$ follows from this  and 
 (\ref{id.union.eqn}) for $ \id_{sp}$, thus proving both versions of (\ref{id.seq.multiply}).  
 (\ref{id.para.multiply}) follows immediately from the definition of a (sp-)ideal and their closure properties.

We prove  (\ref{odot.eqn}) as follows. Suppose a pomset  $r \in \id_{sp} ( L \parallel L')$. Thus $r$ is representable by a labelled  partial order $(V \cup V' \le, \mu) $ for  pomsets $q,q'$ defined by disjoint vertex sets $V,V'$ that are ideals of pomsets  lying in $L,L'$ respectively.
By Theorem \ref{2606y.monk.N-free.thm} applied to $r$, the pomsets $q,q'$ are N-free; hence again by this Theorem, $q,q' \, \in  \id_{sp} (  L) $,  $\id_{sp} (  L')$ respectively. Thus $r \in \id_{sp} (  L) \odot  \id_{sp} ( L')$. We have shown that $  \id_{sp} ( L \parallel L')  \subseteq  \id_{sp} (  L) \odot  \id_{sp} ( L')$, and clearly equality holds. 
\qed
\end{proof}

The set of pomset ideals is not a sub-bi-Kleene algebra of the set of pomset languages, since 
 if the commutative Kleene operations are defined as given by Proposition \ref{prop.powerset.kleene.natural}, then the parallel product of two pomset ideals is not usually an ideal; an analogous statement holds for  
  sp-ideals. However, by taking the ideal closure, or sp-ideal closure, respectively, of the pomset languages defined in the usual way by $ \parallel$ and $^{(*)}$, we obtain bi-Kleene algebras of pomset ideals and sp-pomset ideals. 

\begin{thm}[bi-Kleene algebras of pomset ideals and sp-ideals] \label{thm.ideal.biKleene.homo}
Let $ \Sigma$ be an alphabet. Then $ \id (2^{\pom( \Sigma)})$ is a bi-Kleene algebra provided that the Kleene operations $ 0,1, +, \cdot, ^*$ are interpreted as indicated in Proposition \ref{prop.powerset.kleene.natural} and the commutative Kleene operations $ \parallel, \, ^{(*)}$ are interpreted as 
\begin{equation}
 (I,I') \mapsto \id(I \parallel I') \,   
\textit{ and } \, I \mapsto \cup_{j \ge 0}  \id (I^{(j)}  ) \label{id.interp.eqn}
\end{equation}
respectively.

Furthermore, the set   $ \id_{sp} (2^{\pom( \Sigma)})$ of sp-ideals with labels in $\Sigma$  is a bi-Kleene algebra provided that the Kleene operations $ 0,1, +, \cdot, ^*$ are interpreted as indicated in Proposition \ref{prop.powerset.kleene.natural},  and the commutative Kleene operations $ \parallel, \, ^{(*)}$ are interpreted as 
\begin{equation}
 (I,I') \mapsto \id_{sp}(I \parallel I') \,   
\textit{ and } \,   I \mapsto \cup_{j \ge 0}  \id_{sp} (I^{(j)}  )  \label{id.sp.interp.eqn}
\end{equation}
respectively. 

Lastly, the function 
$$  \id_{sp} (2^{\pom( \Sigma)}) \to  \id (2^{\pom( \Sigma)}) $$
defined by 
\begin{equation}
L \mapsto \id(L) \label{id.func.eqn}
\end{equation}
is an injective bi-Kleene homomorphism.
\end{thm}

\begin{proof}
We first consider $ \id (2^{\pom( \Sigma)})$. Since this set is closed under the Kleene operations $ 0,1, +, \cdot, ^*$, it is a Kleene subalgebra of $2^{\pom( \Sigma)}$. Thus it remains to prove the validity of the bi-Kleene axioms mentioning $ \parallel $ and $ ^{(*)}$.  Associativity of $ \parallel$ follows since for $ I', I'', I''' \in   \id (2^{\pom( \Sigma)})$, 
$$ \id(I' \parallel \id(I'' \parallel I''')) =  \id(I' \parallel I'' \parallel I''') =  \id( \id(I' \parallel I'') \parallel I'''     ) $$ by (\ref{id.para.multiply}) in Lemma \ref{oslash.852t.para.induct.lem}.   The remaining axioms involving only $ 0,1, +, \parallel$ are clear.
The identities in (\ref{eqn.kleene.star.add}) for $ \parallel, ^{(*)}$ follow since for $I \in   \id (2^{\pom( \Sigma)})$,
\begin{align*}
 &   \cup_{j \ge 0}  \id (I^{(j)}  ) && = \,  \cup_{j \ge 1}  \id (I^{(j)}  ) \, \cup \{ 1\} 
   \\
 &  = \, 
 \cup_{j \ge 0} \id(I \parallel  I^{(j)}  ) \, \cup \{ 1\}  
  && =  \,
 \cup_{j \ge 0} \id(I \parallel  \id (I^{(j)}  )) \, \cup \{ 1\} \text{ by (\ref{id.para.multiply})}
 \\  
 & = \,  \id \big( \cup_{j \ge 0}(I \parallel  
  \id ( I^{(j)}  ) ) \big) \cup \{ 1 \}  && \text{ by (\ref{id.union.eqn})}
  \\
&  = \, \id \big(I \parallel  \cup_{j \ge 0} 
  \id (I^{(j)}  ) \big)  \cup \{ 1 \}
 &&      \text{ since  $\parallel$ distributes over unions}   \\
 & =\,  \id(I \parallel  \cup_{j \ge 0} 
  I^{(j)}   ) \, \cup \{ 1 \}   &&  \text{ by (\ref{id.union.eqn}) and (\ref{id.para.multiply}) }
    \\
 &   = \,  \id(I \parallel 
  I^{(*)}   )  \, \cup \{ 1 \}.
  \end{align*} 
  The induction axiom $ s \parallel t \le t \Rightarrow s^{(*)} \parallel t \le t$ follows since for  $I, J \in  \id (2^{\pom( \Sigma)})$  
  \begin{align*}
  &  \id(I \parallel J)  \subseteq  J  \,     \Rightarrow \, I \parallel J  \subseteq  J \, \Rightarrow \,
   I^{(*)} \parallel J \, \subseteq  J = \id(J) \\
   & \Rightarrow \, \id (  I^{(*)}   \parallel J )  \subseteq  J  .
\end{align*}   
 The corresponding result for  $ \id_{sp} (2^{\pom( \Sigma)})$ is proved analogously. 
 We now show that the  function given by (\ref{id.func.eqn})   is a bi-Kleene homomorphism. For  the Kleene operation $  +$  this follows from (\ref{id.union.eqn}). For $ \parallel$, observe that
  $ \id( \id_{sp} (I \parallel I') )= \id  (I \parallel I' )=   
   \id ( \id(I) \parallel \id(I'))$ using (\ref{id.para.multiply}),  and the case of $ ^{(*)}$ then follows from (\ref{id.union.eqn}). 
  The cases of $\cdot$  and 
  $ ^*$  are given by (\ref{id.seq.multiply}). 
 To show injectivity, observe that there is a partial inverse function 
  $$ L \mapsto \id_{sp}(L), $$
  since if $I \in \id_{sp} (2^{\pom( \Sigma)})$ then $ \id_{sp}( \id (I))= \id(I) \cap \pom_{sp}(\Sigma) =\id_{sp}(I) =I$ holds. \qed
\end{proof}

Restricting the homomorphism from $\id_{sp}(2^{\pom(\Sigma)})$ to $\id(2^{\pom(\Sigma)})$ given by (\ref{id.func.eqn}) to the subalgebra of $\id_{sp}(2^{\pom(\Sigma)})$ generated by the set of singleton pomsets $ \big\{ \{ \sigma\} \big\vert \, \sigma \in \Sigma \big\}$ gives an isomorphism onto the subalgebra of $\id(2^{\pom(\Sigma)})$ generated by this set. 

\begin{thm} \label{thm.id.subalgbra.iso.bi-Kleene}
Let $ \Sigma$ be an alphabet. Then the bi-Kleene algebras 
\begin{equation}
\big\{ \id_{sp}( \lb t \rb ) \big\vert \, t \in T_{bi-KA}(\Sigma) \big\} \text{ and }  
\big\{ \id( \lb t \rb ) \big\vert \, t \in T_{bi-KA}(\Sigma) \big\} ,  \label{eqn.Sigma.generate.biKleene}
\end{equation}
with the operations $ \parallel, ^{(*)}$ interpreted as given in (\ref{id.sp.interp.eqn}) and (\ref{id.interp.eqn}) respectively, and the Kleene operations $ 0,1, +, \cdot, ^*$ interpreted as given in Proposition \ref{prop.powerset.kleene.natural}, are isomorphic; an isomorphism is given by 
\begin{equation*}
\id_{sp}( \lb t \rb ) \mapsto \id( \lb t \rb ).
\end{equation*} 
\end{thm}

\begin{proof}
Immediate from Theorem \ref{thm.ideal.biKleene.homo}. \qed
\end{proof}

\begin{prop} \label{prop.id.exch.satisfy}
The classes of pomset ideals and sp-ideals, 
with the operation $ \parallel$ interpreted as in  (\ref{id.interp.eqn}) and (\ref{id.sp.interp.eqn}) respectively, satisfy the exchange law (\ref{exch.eqn}).
\end{prop}

\begin{proof}
We  consider the class of pomset ideals; the proof for the case of sp-ideals is analogous, with $ \id$ replaced by $ \id_{sp}$.  
Let $u,v,x,y \in \id(2^{\pom})$ and suppose $p_u \in u$ and similarly for $ p_v, p_x, p_y$. 
Then the pomset 
$(p_u \parallel p_v)\cdot ( p_x \parallel p_y) \in \id ( p_v \cdot p_y   \parallel p_u \cdot p_x)$ holds, 
and hence 
 \begin{equation*}
 (u \parallel v)\cdot ( x \parallel y) \subseteq   \id( v \cdot y   \parallel u \cdot x )  \label{eqn.id.set.exch} 
\end{equation*}
 follows.
Applying $ \id $ to the left side by using (\ref{id.seq.multiply}) gives $  \id(u \parallel v)\cdot \id( x \parallel y) \subseteq   \id( v \cdot y   \parallel u \cdot x )   $ and thus  (\ref{exch.eqn}) holds. \qed
\end{proof}

\begin{defn}[The  $=_{EX} $ relation] \label{defn.term.EX.reln} \rm
Let $t,t' \in T_{bw-Rat}(\Sigma)$ for an alphabet $ \Sigma$. We say that   $ t =_{EX}  t'$ if $t=t'$ holds in every bw-rational algebra in which the exchange law (\ref{exch.eqn}) also holds.
  We also  define the partial ordering  $ \le_{EX}$ by analogy with 
  (\ref{prec.abbrev.eqn}). 
  
  In view of Theorem \ref{thm.bwrat.equal.imply.bw-equiv}, we will broaden the use of the relations $ \le_{EX}$ and $ =_{EX}$. 
 Clearly $ \le_{bw-Rat} \, \subseteq \, \le_{EX}$ holds, and we exploit this by allowing bw-rational pomset languages to occur in the arguments of $ \le_{EX}$ and $ =_{EX}$; for example, $L  =_{EX} t'$ for term $t'$ and language $L$ if   $ t =_{EX} t'$ holds for at least one (and hence every) term $t \in T_{bw-Rat}(\Sigma)$ satisfying   $ L = \lb t \rb $. 
\end{defn}

\subsection{Summary of proof of our main theorems on pomset ideals}

 \label{subsect.summary.proof.ideal.thms}

The reader is advised to study the proof of Theorem \ref{3960h.mice.rational.exchange.derive.thm}, our last main theorem,  in order to have an insight into the purpose of the lemmas and theorems preceding  it. This proof is straightforward if it is assumed that for any bw-rational term  $t $, the language  
 $ \id_{sp}(\lb t\rb)$ is bw-rational and satisfies 
$  \id_{sp} (\lb t\rb) \le_{EX}  t $. This is precisely the content of Theorem \ref{3696r.trex.ideal.rational.thm}, which is proved by induction on the structure of $t$. The only non-trivial case in this proof is that where $t$ is a parallel product; $t= r_1 \parallel r_2 $, which implies $ \id_{sp}(\lb t\rb)=  \id_{sp}(\lb r_1\rb) \odot  \id_{sp}(\lb r_2\rb)$ by (\ref{odot.eqn}) in Lemma \ref{oslash.852t.para.induct.lem}. Thus it is necessary to prove  Theorem \ref{3696r.trex.ideal.rational.thm} for the special case that $\lb t \rb$ is a $ \odot$-product 
of two  bw-rational ideal  languages. This is implied by Lemma \ref{thm.odot.preserve.bwRat}, which states that for bw-rational terms $r_1,r_2$, $\lb  r_1 \rb  \odot \lb   r_2  \rb \le_{EX} r_1 \parallel r_2 $ holds.  Its proof is by induction on the sum of the widths of $r_1 $ and $r_2$ and entails proving that for each $k \ge 1$,  $  \big( \lb  r_1 \rb  \odot \lb   r_2  \rb   \big)\cap \para_k \le r_1 \parallel r_2$ holds.  The cases $ k \ge 2$ can be inferred from the case $k =1 $ using the inductive hypothesis and Corollary \ref{2058f.turkey.fun.cor}. The case $k=1$ follows from Corollary \ref{cor.odot.into.reg.over.para}, which shows that for bw-rational terms $L_1, L_2$, the language $(L_1 \odot L_2) \cap \seq$ is definable by a regular term with $\odot$-product languages substituted for its ground terms.

\section{Two automata-theoretic lemmas}

\label{sect.auto.lems}

In order to prove our main theorems, we need the following automata-theoretic results.

\begin{lem} \label{4801j.santa.auto.lem}
Let $ \Gamma$ be a finite  alphabet and let $  L$ be a regular language over $ \Gamma$.   
Let $ \approx$ be a congruence  of finite index of the monoid $( \Gamma^*  , 1, \cdot)$ and assume that $L$ is a union of some of the  $ \approx$-congruence classes.  Define 
a finite set $\Delta$ and a function $ \theta : \Delta \to \Gamma^*/ \approx $. 
Define  the set 
 \begin{align*} V=  & \mathlarger{\{} \\
  &  \delta_{1}  \ldots   \delta_{b}   \vert \, b \ge 0, \\ 
                         &  \text{each } \delta_{i} \in \Delta , \\
   &     \theta(\delta_{1}) \ldots  \theta(\delta_{b})  \subseteq   L \\
    &  \mathlarger{\}}. 
    \end{align*} 
 Then $V$ is a regular language over $ \Delta$.
\end{lem}

\begin{proof}
We define a deterministic finite state automaton $B$ as follows. $B$ has state set  $\Gamma^*/ \approx $. Its initial state is the $ \approx$-class containing $1$,  and its final states are those whose union is $L$. For each $ \delta
 \in \Delta$, $B$ has  a binary transition relation  $\underset{\delta }{\leadsto}$ on $\Gamma^*/ \approx $ as follows; 
for  $S \in \Gamma^*/ \approx $, 
we define $ S \underset{\delta }{\leadsto} S'$, where $S \theta( \delta) \subseteq S'$. 
Since $ \approx$ is a congruence,  the class  $S'$ exists and is uniquely determined by $S$ and $ \delta$. 

Let $S_0$ be the initial state of $B$, so $ 1 \in S_0$.  
Given $ \delta_1, \ldots, \delta_b \in \Delta$ for $ b \ge 0$,  there are states $S_1, \ldots , S_b$ of $B$ such that 
$ S_0  \underset{\delta_1 }{\leadsto} S_1  \underset{\delta_2 }{\leadsto} \ldots \underset{\delta_b }{\leadsto} S_b$ holds. Thus  $  S_b \subseteq L \iff   S_0 \, \theta(\delta_{1}) \ldots  \theta(\delta_{b})  \subseteq   L    \iff  \theta(\delta_{1}) \ldots  \theta(\delta_{b})  \subseteq   L   \iff    \delta_1 \ldots \delta_b \in V $, proving that $B$  accepts $V$. 
\qed
\end{proof}

\begin{lem} \label{2960u.elk.cartesian.lang.lem}
Let $ \Gamma$ be a finite  alphabet and let $  L_1, L_2$ be regular languages over $ \Gamma$.   
Let $ \approx$ be a congruence  of finite index of the monoid $( \Gamma^*  , 1, \cdot)$ and assume that each  $L_i$ is a union of some of the  $ \approx$-congruence classes.  Define 
a finite set $\Delta$ and a function $ \theta : \Delta \to \Gamma^*/ \approx $,  
 and define the language 
 \begin{align*}
  U= \;\;\; & \mathlarger{\{} \\
  & ( \delta_{11},  \delta_{21} ) \ldots  ( \delta_{1b},   \delta_{2b} ) \vert \, b \ge 0, \\ 
&  \text{each } \delta_{ij} \in \Delta , \\
  &     \theta(\delta_{i1}) \ldots  \theta(\delta_{ib})  \subseteq   L_i  \text{ for each } i =1,2 \\
   &  \mathlarger{\}} 
\end{align*}
over $ \Delta \times \Delta$. 
Then $U$ is regular.

\end{lem}

\begin{proof}
For each $j = 1,2$, let $V_j$ be the language defined as $U$ is, but satisfying only the condition  $\theta(\delta_{i1}) \ldots  \theta(\delta_{ib})  \subseteq   L_i $ for $i =j$.  Thus $U = V_1 \cap V_2$. It suffices thus to prove that each $V_j$ is regular, and this follows from Lemma  \ref{4801j.santa.auto.lem}, since  regularity is preserved by substitution. 
\qed
\end{proof}

\section{Expressing the sequential sublanguage of a $\odot$-product as a regular function of `smaller'  non-sequential sublanguages of  $\odot$-products }

Our main result in this section is  Corollary \ref{cor.odot.into.reg.over.para}, which is an essential intermediate result for proving that the $ \odot$ operation preserves bw-rationality, 
as indicated in Section \ref{subsect.summary.proof.ideal.thms}. 


\begin{lem} \label{lem.language.overCi.equals.P}
Let $\Sigma$ be an alphabet and let   $C_1,\ldots, C_m \in\pom_{sp}(\Sigma) $ be non-sequential  and let  $ \{\gamma_1, \ldots, \gamma_m \} $ be an alphabet and  for $i =1,2$  let 
$$L_i=L_i( \gamma_1, \ldots, \gamma_m)   \subseteq  \{ \gamma_1, \ldots, \gamma_m \}^* . $$ 
 Let $\Delta$ be a  set and let $ \approx$ be a congruence  on the monoid $(\{\gamma_1, \ldots, \gamma_m \}^*,1, \cdot)$ such that each language  $ L_i $ is a union of $\approx$-classes.
  Let $\phi$ be a function from $ \Delta$ onto $\{\gamma_1, \ldots, \gamma_m \}^* / \approx$. Define the language 
  \begin{align*}
  U \, = \, & \mathlarger{\{} \\
  & ( \delta_{11},  \delta_{21} ) \ldots  ( \delta_{1b},   \delta_{2b} ) \vert \, b \ge2, \\ 
&  \text{each } \delta_{ij} \in \Delta , \\
  &     \phi(\delta_{i1}) \ldots  \phi(\delta_{ib})  \subseteq   {L}_i    \text{ for each } i =1,2 \\
   &  \mathlarger{\}}
\end{align*}  
over the alphabet $\Delta \times \Delta$, 
and write
 \begin{equation*}
\tilde{U} =  U\big( (\delta_1, \delta_2) \mathlarger{\setminus} \big(  \phi(\delta_1)( C_1, \ldots , C_m)   \odot  \phi(\delta_2)(C_1,\ldots, C_m)  \big)  \cap (\para \cup \Sigma) \big\vert \; (\delta_1, \delta_2) \in \Delta \big),
\end{equation*} 
where for each $ \delta \in \Delta$, we write $ \phi(\delta)( C_1, \ldots , C_m)$ to denote the language $\phi(\delta)$ with each letter $ \gamma_i$ replaced by the language $  C_i $. 
 Then
\begin{equation*}
\big( L_1 (C_1,\ldots, C_m)  \odot     L_2 (C_1,\ldots, C_m)    \big) \cap \seq = \tilde{U}
\end{equation*}
 holds.

\end{lem}

\begin{proof}  
Let  $p \in \big(  L_1 (C_1,\ldots, C_m)  \odot      L_2 (C_1,\ldots, C_m)    \big) \cap \seq $. We will show that $ p \in \tilde{U}$.   We have
 \begin{equation}
  p= p_1 \ldots p_b \label{eqn.p.decomp.seq}
\end{equation} 
  for $b \ge 2$ and each $ p_j \in \para \cup \Sigma$ and  there are pomsets 
$$q_i \in  L_i (C_1,\ldots, C_m)  $$ 
such that $$ p \in   q_1 \odot   q_2   .$$
 Clearly 
  each 
  $q_i = r_{i1} \ldots r_{ia_i}$ for some $a_i \ge 0$, where each pomset $r_{ij} \in \cup_{k=1}^m C_k  $  and is hence non-sequential. Hence the vertices in any pomset $r_{ij}$ all lie in one of the pomsets $p_l$  and since their ordering in each $q_i$ is preserved in $p$, 
  for any $j < j'$   the vertices in  $r_{ij}$ and $r_{ij'}$  lie in $p_l$ and $p_{l'}$ respectively for some $l \le l'$. Hence by gathering together adjacent pomsets $r_{ij}$ whose vertices lie in the same pomset  $p_l$,  we may write  
  $$q_i = w_{i1} \ldots w_{ib}$$
   where each $w_{ij}$ is a sequence of pomsets all lying  in  $\cup_{k=1}^m C_k $ and the vertices in $w_{ij}$ occur in the pomset $p_j$. Thus each
   \begin{equation*}
    p_j \in \,  w_{1j} \odot  w_{2j}  \,  \cap \, (\para \cup \Sigma) 
\end{equation*}   
   holds.

     Clearly each language
 $$   L_i (C_1,\ldots, C_m) = \bigcup_{ \subalign{
 & v \,  = v( \gamma_1, \ldots , \gamma_m) \, \in  {L}_i \\  
 &    } 
  }
    v(C_1,\ldots, C_m) 
    $$
    and so each $ w_{i1} \ldots w_{ib} = q_i \in v_i(C_1,\ldots, C_m) $ for words
    $ v_i = v_i( \gamma_1, \ldots , \gamma_m) \in  
    {L}_i  $. Since the pomsets in the language $\cup_{k=1}^m C_k $  are non-sequential,   by Part (2) of Lemma \ref{lem.pomset.decomp.unique} there are words 
  $  v_{ij} =  v_{ij}(\gamma_1, \ldots, \gamma_m)$ such that $v_i = v_{i1} \ldots v_{ib}$ and 
  $ w_{ij} \in  v_{ij}(C_1,\ldots, C_m)   $ holds. Hence  
  \begin{equation}
   v_{i1} \ldots v_{ib} \in   {L}_i \label{eqn.L-i.and.v-ij.contain}
\end{equation}   holds.

  Since $ \phi$ is onto, we may suppose each $v_{ij} \in   \phi(\delta_{ij})  $ for $ \delta_{ij} \in \Delta$. Then  each 
   $ w_{ij} \in   \phi( \delta_{ij}) (C_1,\ldots, C_m)   $ and so 
   \begin{equation}  p_j \in  \,
       \phi( \delta_{1j})(C_1, \ldots, C_m)  \odot   \phi( \delta_{2j})(C_1, \ldots, C_m) \, \cap \,  (\para \cup \Sigma)  \label{eqn.odot.cap.Para.and.Sigma} 
    \end{equation}
  holds. Also,  $ \phi( \delta_{i1}) \ldots  \phi( \delta_{ib})  \subseteq W_i $ for some $\approx$-class $W_i$, since $\approx$ is a congruence. From (\ref{eqn.L-i.and.v-ij.contain}), $L_i \cap W_i \neq \emptyset$ holds, and so  $ \phi( \delta_{i1}) \ldots  \phi( \delta_{ib})  \subseteq L_i $ follows since each $L_i$ is a union of $\approx$-classes.   
   Hence $p \in \tilde{U}$ follows from (\ref{eqn.p.decomp.seq}) and (\ref{eqn.odot.cap.Para.and.Sigma}). Thus we have proved $  \big(  L_1 (C_1,\ldots, C_m)  \odot      L_2 (C_1,\ldots, C_m)    \big) \cap \seq \subseteq  \tilde{U}$.

 Conversely, suppose  that $ p \in \tilde{U}$ holds.  
 Then   there exist  $b \ge 2$ and elements $ \delta_{ij} \in \Delta$ such that 
  \begin{align*}
 &    \phi(\delta_{i1}) \ldots \phi( \delta_{ib})  \subseteq  {L}_i(\gamma_1, \ldots, \gamma_m) 
  \text{ for each } i =1,2   \\
& \text{and } 
  p = p_1\ldots p_b, \text{ where}   \\
  & \text{each } p_j \in \big(  \phi( \delta_{1j})(C_1, \ldots, C_m)  \odot    \phi( \delta_{2j})(C_1, \ldots, C_m) \big) \cap (\para \cup \Sigma)   
  \end{align*} 
  and so there exist words $ v_{ij} = v_{ij}( \gamma_1, \ldots, \gamma_m) \in   \phi( \delta_{ij})  $ such that 
  each $$ p_j \in \big(  v_{1j}(C_1, \ldots, C_m)  \odot   v_{2j}(C_1, \ldots, C_m) \big) \cap (\para \cup \Sigma)$$ and hence there exist pomsets 
  $$w_{ij} \in  v_{ij}(C_1, \ldots, C_m) $$ such that 
  each $ p_j \in \big(  w_{1j}  \odot  w_{2j}  \big)   \cap (\para \cup \Sigma)$. 
   Thus 
  \begin{equation*}
  p \in  \prod_{j=1}^b \big( ( w_{1j}  \odot  w_{2j} )    \cap (\para \cup \Sigma)  \big)
 \, \subseteq 
  \, (  w_{11} \ldots w_{1b}  \odot      w_{21} \ldots w_{2b}   ) \cap \seq 
\end{equation*} 
   where for each $ i =1,2$, clearly 
 $   v_{i1} \ldots v_{ib}  \in   {L}_i(\gamma_1, \ldots, \gamma_m)     $ and thus    
 \begin{align*}
  w_{i1} \ldots w_{ib} &\in   v_{i1}(C_1, \ldots, C_m)  \ldots v_{ib}(C_1, \ldots, C_m)   \\
  & \subseteq   \phi( \delta_{i1}) (C_1,\ldots, C_m) \ldots \phi( \delta_{ib})(C_1,\ldots, C_m)   \\ & \subseteq   {L}_i(C_1,\ldots, C_m) 
\end{align*} 
  holds.   Thus $ p \in \big( L_1 (C_1,\ldots, C_m)  \odot     L_2 (C_1,\ldots, C_m)    \big) \cap \seq$ follows, as required. 
     \qed 
\end{proof}

The main result of this section follows.

\begin{cor} \label{cor.odot.into.reg.over.para}
Let $ \Sigma$ be an alphabet and let  terms $c_1,\ldots, c_m \in T_{bw-Rat} ( \Sigma) $ be non-sequential with each $ \lb c_j \rb = C_j$  and let 
$$t_1=t_1( \gamma_1, \ldots, \gamma_m),  t_2= t_2( \gamma_1, \ldots, \gamma_m)$$
 be regular terms over an alphabet
 $ \{\gamma_1, \ldots, \gamma_m \}$. 
 Let $\Delta$ be a finite set and let $ \approx$ be a congruence of finite index on the monoid $(\{\gamma_1, \ldots, \gamma_m \}^*,1, \cdot)$ such that each language  $ \lb t_i \rb$ is a union of $\approx$-classes.
  Let $\phi$ be a function from $ \Delta$ onto $\{\gamma_1, \ldots, \gamma_m \}^* / \approx$. 
  Then there is a regular term $u$ over the alphabet $ \Delta \times \Delta$ such that 
\begin{align}
&  t_1 (C_1,\ldots, C_m)  \odot      t_2 (C_1,\ldots, C_m)    \,  \cap   \,\seq \nonumber
 \\
& = \nonumber \\
&  u\big( (\delta_1, \delta_2) \mathlarger{\setminus}  \phi(\delta_1)(  C_1,\ldots, C_m)   \odot  \phi(\delta_2)(C_1,\ldots, C_m)    \cap (\para \cup \Sigma) \big\vert \; (\delta_1, \delta_2) \in \Delta \big) \label{eqn.reg.over.odot.intersect.not_seq}
 \end{align}
 holds and for each word $ ( \delta_{11},  \delta_{21} ) \ldots  ( \delta_{1b},   \delta_{2b} ) \in \lb u \rb $ and  $ i =1,2$,
 \begin{equation}
   \phi(\delta_{i1}) \ldots  \phi(\delta_{ib})  \subseteq \lb  {t}_i(\gamma_1, \ldots, \gamma_m) \rb .
   \label{eqn.delta.without.C-j.into.t-i}   
 \end{equation}
 Also, 
   for each  $(\epsilon_1, \epsilon_2 )  \in \supp(u)$ and $i =1,2$,  
 \begin{equation}
  \width \big( \phi(\epsilon_i)( C_1, \ldots, C_m) \big) \le      \width \big( t_i( C_1, \ldots, C_m) \big) \label{eqn.width.phi.less.than.t-i}
\end{equation} 
 holds. 
\end{cor}

\begin{proof}

By  Lemma \ref{2960u.elk.cartesian.lang.lem}, there is a term $u \in T_{Reg} ( \Delta \times \Delta)$ such that 
  \begin{align*}
  \lb u \rb =   & \, \mathlarger{\{}  \\
    &  ( \delta_{11},  \delta_{21} ) \ldots  ( \delta_{1b},   \delta_{2b} ) \vert \, b \ge2, \\ 
&  \text{each } \delta_{ij} \in \Delta , \\
  &     \phi(\delta_{i1}) \ldots  \phi(\delta_{ib})  \subseteq \lb  {t}_i(\gamma_1, \ldots, \gamma_m) \rb    \text{ for each } i =1,2 \\
   &  \mathlarger{\}}
\end{align*}
and so (\ref{eqn.reg.over.odot.intersect.not_seq}) holds by Lemma \ref{lem.language.overCi.equals.P}, with $ \lb t_i \rb $ in the role of the languages  $L_i$, and (\ref{eqn.delta.without.C-j.into.t-i}) holds from the definition of $u$.

To prove (\ref{eqn.width.phi.less.than.t-i}), let     $(\epsilon_1, \epsilon_2 ) \in \supp(u)$. Thus  there is  a word \\
 $  ( \delta_{11},   \delta_{21} ) \ldots   ( \delta_{1b},  \delta_{2b} )  \in \lb u \rb $  in which  $(\epsilon_1, \epsilon_2 )$ occurs and so there exists $ b' \le b$ such that for each 
 $j  =1,2$,   $  \, {\epsilon_j} = \delta_{jb'}$ holds and so from (\ref{eqn.delta.without.C-j.into.t-i}), 
   \begin{equation}
     \phi(\delta_{j1} )(C_1, \ldots, C_m) \ldots \phi(\delta_{jb})(C_1, \ldots, C_m)  \subseteq {t}_j(C_1, \ldots, C_m)         \label{eqn.deltas.into.t-i}
\end{equation}  
  holds. 
    Since the  elements of $ \phi(\Delta)$ are   non-empty sublanguages of $\{ \gamma_1, \ldots , \gamma_m \}^*$, each pomset language  $\phi(\delta_{jk} )(C_1, \ldots, C_m)$ is also non-empty and so 
     $$\width \big( \phi(\epsilon_j)( C_1, \ldots, C_m) \big) \le \width \big( \phi(\delta_{j1} )(C_1, \ldots, C_m) \ldots \phi(\delta_{jb})(C_1, \ldots, C_m) \big)$$
    holds. Thus  (\ref{eqn.width.phi.less.than.t-i}) follows from (\ref{eqn.deltas.into.t-i}). 
     \qed

\end{proof}

Lemma \ref{2360d.hart.gather.lem} relates the $ \odot$-product of two  languages defined as parallel products to  $ \odot$-products of  their respective parallel components.

\begin{lem}\label{2360d.hart.gather.lem} 
 Let $L_1, L_2 $ be bounded-width languages of sp-pomsets over an alphabet $ \Sigma$, where each 
 $$ L_i = S_{i1} \parallel \cdots\parallel  S_{im_i}  $$
 for $ m_i \ge 1$ and 
  pomset languages $S_{ij}$ satisfying $  S_{ij}  \subseteq \seq \cup \Sigma $.  Let $ k \ge 1$. Then
  the following holds; 
   the language $ ( L_1 \odot L_2 )   \cap \para_k$ is the union of all languages of the form 
 $  M_1 \parallel \cdots \parallel M_k   $, where each language 
 $$  M_j  =      ( \underset{b \in T_{1j}}{ \parallel} \!\!\!   S_{1b} \;\; \odot \;   \underset{b \in T_{2j}}{ \parallel} \!\!\! S_{2b}  \; ) \cap  (\seq \cup \Sigma)  ,   $$
where for each $i \in \{1,2\}$, sets
    $ T_{i1}, \ldots, T_{ik}$ partition the set  $  \{ 1, \ldots, m_i\}$, and   such that for each $j \le k$, the  set $T_{1j}  \cup T_{2j} \not= \emptyset $.  If additionally $ k \ge 2$ holds then 
   $$ \sum_{i=1}^2  \width(  \parallel_{b \in T_{ij} }   S_{ib}) <    \sum_{i=1}^2  \width (L_i )$$
  holds for each $ j \le k$.

\end{lem}

\begin{proof}
Observe first that each metaterm  $S_{ij}$ occurs exactly once in both  the expressions $ L_1 \odot L_2$ and $  M_1 \parallel \cdots \parallel M_k   $ under the  conditions on the sets $T_{ij}$ given in   the Lemma.   Additionally, if any ordering occurs in  a pomset in  $  M_1 \parallel \cdots \parallel M_k   $ between vertices in a pomset in $S_{ij}$ and in $S_{i'j'}$, then either  $i= i' \wedge j=j'$ or $ i \not= i'$ holds, and hence the same ordering can occur in $ L_1 \odot L_2$. Thus any pomset in a language $  M_1 \parallel \cdots \parallel M_k   $
  under the given conditions lies in $ (L_1 \odot L_2) \cap \para_k$. 
  
Conversely, let $ p \in  (  L_1 \odot L_2) \cap \para_k $. Thus each language $S_{ab} $ contains a pomset $q_{ab}$ such that the vertex set of $p$ is the pairwise disjoint union of all the vertex sets of the pomsets $q_{ab}$, with the same labelling and the same vertex ordering within each $q_{ab}$.

 Write  $p = p_1 \parallel \ldots \parallel p_k$, with each $p_j \in \seq \cup \Sigma$. Given any  $j \le k$ and $i \in \{ 1,2 \}$, let $T_{ij}$ be the set of elements of $ \{1, \ldots, m_i\} $ such that $p_j$ contains at least one vertex from $  q_{ib}  $ if and only if $b \in T_{ij}$.  Since each  $  q_{ib} \notin \para  $, and its ordering is preserved in $p$, all vertices  in  $  q_{ib}  $ occur in $p_j$ if $b \in T_{ij}$. Hence 
$ j \not= j' \Rightarrow T_{ij} \cap T_{ij'} = \emptyset$ follows, and since  the vertices of every pomset $q_{ib}$ must occur in some $p_j$, $ \{1, \ldots, m_i\} = \cup_{j=1}^k  T_{ij} $ holds, proving 
 the partitioning property of the sets $T_{ij}$ asserted by the Lemma.

Since the ordering of the vertices in $ \parallel_{ b\le m_i} q_{ib}$ and hence in $ \parallel_{ b \in T_{ij}} q_{ib}$ is preserved in $p$, and  each $p_j \in \seq \cup \Sigma$, it follows that  $p_j \in M_j $, with $ M_j $ defined as in the statement of the Lemma using the sets $T_{ij}$. The  assertion that  $T_{1j}  \cup T_{2j} \not= \emptyset $ holds follows since $p \in \para_k$ and each $ M_j \not= \{ 1 \}$. 

The width property asserted by the Lemma holds if  $ k \ge 2$ since for each $ j\le k$ and  $ i \in \{1,2\} $, $T_{ij} \subseteq \{ 1, \ldots , m_i \} $ and so 
$$  \width(  \parallel_{b \in T_{ij} }   S_{ib}) \le   \width(  S_{i1} \parallel \cdots \parallel S_{im_i} ) =   \width ( L_i )$$
 holds, with strict inequality for at least one $ i \in \{1,2 \}$, since given any $ j, j' \le k$ with $j' \not= j$,   $T_{ij'} \not= \emptyset$ holds for at least one element $ i \in \{1,2\}$, and so for  every $ b \in T_{ij'} =  T_{ij'} - T_{ij} $, the term $S_{ib'}$ occurs in the middle term but not on the left side of 
 the above inequality, and $ S_{ib'} \ni q_{ib} \not= \{ 1 \}$. Thus we have proved the Lemma. 
 \qed
\end{proof}

Corollary \ref{2058f.turkey.fun.cor} gives an inductive step in the proof of Theorem \ref{thm.odot.preserve.bwRat}, our third main theorem.

\begin{cor} \label{2058f.turkey.fun.cor}
Let $L_1, L_2 $ be  bw-rational languages of sp-pomsets over an alphabet $ \Sigma$, and assume that for any bw-rational languages $L_1', L_2'$ satisfying $\sum_{i=1}^2 \width(L_i') < \sum_{i=1}^2 \width(L_i)$, the language   $(L_1' \odot L_2') \cap (\seq \cup \Sigma)$ is bw-rational and satisfies 
$$(L_1' \odot L_2') \cap (\seq \cup \Sigma)  \le_{EX} L_1' \parallel L_2'  .$$
 Let $ k \ge 2$. 
Then  $  (L_1 \odot L_2)    \cap \para_k$ is bw-rational and 
 $$  (L_1 \odot L_2)    \cap \para_k \le_{EX} L_1 \parallel L_2$$ holds.
\end{cor}

\begin{proof}
Using the distributive law for $ \parallel$, we may assume that each  $ L_i = S_{i1} \parallel \cdots\parallel  S_{im_i}  $
 for $ m_i \ge 1$ and 
bw-rational   pomset languages $S_{ij}$ satisfying $  S_{ij}  \subseteq \seq \cup \Sigma $.   
  
  We first  prove that  $  M_1 \parallel \cdots \parallel M_k \le_{EX}   L_1 \parallel L_2$ holds, where the languages $M_j$ are as defined   using sets $T_{ij}$ as in Lemma \ref{2360d.hart.gather.lem}; in particular, $\cup_{j=1}^k T_{ij} = \{ 1, \ldots, m_i \}$ for each $ i \in \{ 1,2\}$. From the conclusion of that Lemma and the extra hypotheses assumed here,  each language  $M_j$ is bw-rational and  
 $$ M_j \le_{EX}  ( \underset{b \in T_{1j}}{ \parallel} \!\!\!   S_{1b} ) \; \parallel \; (  \underset{b \in T_{2j}}{ \parallel} \!\!\! S_{2b}  \; ) $$
 holds. 
 Thus 
  $$  M_1 \parallel \cdots \parallel M_k \le_{EX} 
    {\mathlarger{\parallel}}_{j=1}^k \big( ( \underset{b \in T_{1j}}{ \parallel} \!\!\!   S_{1b} ) \; \parallel \; (  \underset{b \in T_{2j}}{ \parallel} \!\!\! S_{2b}  \; ) \big) = L_1 \parallel L_2   $$
    holds. Thus the Corollary follows since there are finitely many ways of defining collections of sets $T_{ij}$ satisfying the conditions given in Lemma \ref{2360d.hart.gather.lem} and so by that Lemma,  $ L_1 \odot L_2    \cap \para_k $ is a finite union of bw-rational languages $R$ satisfying $ R \le_{EX} L_1 \parallel L_2$. 
 \qed
\end{proof}

\section{The main theorems for sp-ideals of rational languages}

\label{sect.main.thms.ideals}

We now show that $ \odot$ preserves bw-rationality of pomset languages, and defines a language that is $ =_{EX}$-equivalent to the parallel product of the languages. We first need an automata-theoretic lemma.

\begin{lem}\label{6299x.reg.lang.equiv.lem}
Let $ \Gamma$ be a finite alphabet, and let $ S$ be a finite set of regular languages over $\Gamma$.
Then there exists a congruence $\approx$ of finite index of the monoid $( \Gamma^*  , 1, \cdot)$
 such that each language  $L\in S$ is the union of  a subcollection of $\approx$-equivalence classes. 
\end{lem}

\begin{proof}
Since the conjunction of two  congruences  of finite index is itself a congruence  of finite index, we may assume that $S$ is a singleton, $S = \{ L \}$.  
Let $A$ be a deterministic finite state automaton accepting the language $ L $. We assume $A$ has state set $Q$ and a binary transition relation $ \underset{w}{\leadsto} \, \subseteq Q \times Q $ for each $w \in \Gamma^*$. For any function $ \theta: Q \to Q$, let $$ K_\theta = \{ w \in \Gamma^* \vert\, q \underset{w}{\leadsto} \theta(q) \, \forall q \in Q \}.$$ 
  For any $w \in \Gamma^*$, there is a function $ \theta: Q \to Q$ such that for any $q \in Q$, there exists a state $ \theta(q)$ satisfying $  q \underset{w}{\leadsto} \theta(q) $, and so $ w \in K_\theta$; furthermore, if any $ w \in K_\theta \cap K_{\theta'}$, then for each $q \in Q$ both  $  q \underset{w}{\leadsto} \theta(q) $ and  $  q \underset{w}{\leadsto} \theta'(q) $ hold. Since $A$ is deterministic, $ \theta(q) = \theta'(q)$ follows and so $ \theta = \theta'$. Thus the sets $K_\theta$ partition $ \Gamma^*$ and are clearly regular, and for each $ \theta, \theta': Q \to Q$, there exists $ \theta''$ satisfying $K_\theta K_{\theta'} \subseteq K_{\theta''}$. Clearly $ L$  is the union of  a collection of languages $K_\theta$; since there are finitely many functions from $Q$ into $Q$,  the Lemma follows. 
 \qed
\end{proof}

Lemma \ref{lem.ordering.reg.over.paral} will be used with Corollary \ref{cor.odot.into.reg.over.para} to transform a `regular over parallel' term given by  $u$ in the statement of this Corollary into a sum of parallel terms in Theorem \ref{thm.odot.preserve.bwRat}, which shows that $ \odot$ preserves bw-rationality.

\begin{lem} \label{lem.ordering.reg.over.paral}
Let $ \{\gamma_1, \ldots, \gamma_m \}$ be an alphabet and let terms
$d_1, \ldots, d_k  \in T_{Reg}( \gamma_1, \ldots, \gamma_m )$ be  such that the languages 
$  \lb d_1 \rb , \ldots, \lb d_k \rb $ partition $ \{ \gamma_1, \ldots, \gamma_m \}^*$  and for each $i,i' \le k $ there is a term $ d_j$ satisfying $ \lb d_i d_{i'} \rb \subseteq \lb d_j \rb$.

 Let $ \Delta$ be a set and let $ \phi : \Delta \to \{d_1, \ldots, d_k\}$ be a bijection.  
Then for any  term $u \in T_{Reg}(\Delta \times \Delta )$, there is a   set $\Lambda \subseteq \Delta \times \Delta $ such that  
$$ u( (\delta, \delta') \setminus \phi( \delta) \parallel \phi(\delta') \vert \; \delta, \delta' \in \Delta)   \le_{EX} \sum_{(\delta,\delta') \in \Lambda }  \phi( \delta) \parallel \phi(\delta') $$
 holds  and for each element  
$ (\epsilon_1 , \epsilon_2) \in \Lambda$, there is a word 
$(\delta_{11},  \delta_{21} ) \ldots   ( \delta_{1b},   \delta_{2b} ) \in \lb u \rb $ such that 
$ \lb \phi( \delta_{i1}) \ldots \phi (\delta_{ib}) \rb \subseteq \lb \phi( \epsilon_i) \rb$ for each $ i \in \{1,2\}$, where if $b=0$ the product 
$ \lb \phi( \delta_{i1}) \ldots \phi (\delta_{ib}) \rb$ is defined to be the language $\{ 1 \}$.

\end{lem}

\begin{proof}
This follows by induction on the structure of $u$.   For convenience, since $ \phi$ is a bijection we may define  $ 1_\Delta \in \Delta$ to be the element satisfying $ 1 \in \lb  \phi(1_\Delta) \rb $, and for any regular  term $x$ over the alphabet  $ \Delta \times \Delta$ we define  $ \bar{x}= x( (\delta, \delta') \setminus \phi( \delta) \parallel \phi(\delta') \vert \; \delta, \delta' \in \Delta)  $.   If $u$ is $0$ or an element of $\Delta \times \Delta $ then $ \Lambda$ is as follows; 
$$ u = \begin{cases}
0 &  \Lambda = \emptyset \\ 
(\delta_1, \delta_2) & \Lambda = \{ (\delta_1, \delta_2) \}
\end{cases}
$$
and the conclusion of the Lemma is immediate. If $u =1$, we define $  \Lambda = \{( 1_\Delta, 1_\Delta)\} $; for then by Theorem \ref{thm.bwrat.equal.imply.bw-equiv},  $ \bar{u} = 1 \le_{bw-Rat}  \phi(1_\Delta) \parallel  \phi(1_\Delta)$ and $\{ 1 \} \subseteq \lb \phi(1_\Delta) \rb$ hold, as required by the Lemma. 
If $u$ is a sum of terms, then the result follows from the inductive hypothesis applied to each of these terms. 

There remain two cases, in both of which it is convenient to define multiplication on the set $ \Delta$ as follows, using the fact that  $ \phi$ is a bijection; if $ \phi( \delta) \phi( \delta') \subseteq \phi( \delta'')$, then $ \delta \delta' = \delta''$ holds. 
  Clearly this definition turns $ \Delta$ into a monoid, with  $ 1_\Delta$ as the    identity element.  Furthermore, for any $ \delta_1, \ldots , \delta_r,  \in \Delta$, 
$$ \phi( \delta_1) \ldots \phi( \delta_r) \subseteq \phi( \delta_1 \ldots \delta_r )$$ 
follows by induction on $r$.    
   Thus  the condition $ \lb \phi( \delta_{i1}) \ldots \phi (\delta_{ib}) \rb \subseteq \lb \phi( \epsilon_i) \rb$ in the statement of the Lemma is equivalent to  $ \delta_{i1} \ldots \delta_{ib} = \epsilon_i$, where if $b =0$  this states that $ 1_\Delta=  \epsilon_i $. 
\begin{itemize}
\item
Suppose $ u = u_1 u_2$. By the inductive hypothesis,  for each $ i \in \{ 1,2\}$, there are    sets $\Lambda_i \subseteq \Delta \times \Delta $ such that
$$ \bar{u_i} \le_{EX} \sum_{(\delta,\delta') \in \Lambda_i }  \phi( \delta) \parallel \phi(\delta'). $$
Define the set 
$$ \Upsilon = 
\big\{ ( \epsilon_{11} \epsilon_{21} , \epsilon_{12}\epsilon_{22}) \vert \,  (\epsilon_{i1},  \epsilon_{i2}) \in  \Lambda_i \text{ for each $ i \in \{1,2 \}$}  \big\} . $$
From the exchange  law and the fact that by Theorem \ref{thm.bwrat.equal.imply.bw-equiv}, $ \phi(\delta) \phi( \delta') \le_{bw-Rat} \phi(\delta \delta')$ always holds, 
\begin{align*}
 \bar{u} \; \le_{EX} \;\;
   \big(  &  \sum_{  \;\;\;\;\;  (\epsilon_{11}, \epsilon_{12}) \in \Lambda_1 \;\;\;\;\;}   \phi( \epsilon_{11}) \parallel
 \phi(\epsilon_{12}) \big) \;  \big( \sum_{(\epsilon_{21}, \epsilon_{22}) \in \Lambda_2 }  \phi( \epsilon_{21}) \parallel \phi(\epsilon_{22}) \big) 
 && \le_{EX} \\
 & \sum_{ \;\;\;\;\;\begin{subarray}{12}
   (\epsilon_{11},  \epsilon_{12}) \in  \Lambda_1, \\
  (\epsilon_{21},  \epsilon_{22}) \in  \Lambda_2 
 \end{subarray} \;\;\;\;\; }  \phi(  \epsilon_{11}) \phi(\epsilon_{21}) \parallel \phi( \epsilon_{12}) \phi (\epsilon_{22})  
&&    \le_{bw-Rat}
   \\
 & \sum_{  \;(\epsilon_{11} \epsilon_{21}, \epsilon_{12}\epsilon_{22}) \in \Upsilon \;}  \phi(  \epsilon_{11}\epsilon_{21}) \parallel \phi( \epsilon_{12}\epsilon_{22})
  && \le_{bw-Rat} 
   \\
    & \sum_{\;\; \; \;\;\;\;\;(\delta,\delta') \in \Upsilon \;\;\;\;\;\; \;\;}  \phi( \delta) \parallel \phi(\delta')
 \end{align*}
holds. 
In addition, if $ \lambda \in \Upsilon$, say $ \lambda = ( \epsilon_{11} \epsilon_{21} , \epsilon_{12}\epsilon_{22})$ with  each $ (\epsilon_{i1},  \epsilon_{i2}) \in  \Lambda_i$,  by the inductive hypothesis
 there is a word 
$w_i = (\delta_{i11},  \delta_{i21} ) \ldots   ( \delta_{i1b_i},   \delta_{i2b_i} ) \in \lb u_i \rb $ such that 
$ \delta_{ij1} \ldots  \delta_{ijb_i} =\epsilon_{ij}$ for each $ i ,j \in \{1,2\}$.
Thus the word 
$ w_1 w_2 \in  \lb u \rb$, and the product of all the $j$th components of the letters of $ w_1 w_2$ is 
$   \delta_{1j1} \ldots  \delta_{1jb_1}   \delta_{2j1} \ldots \delta_{2jb_2} = \epsilon_{1j}  \epsilon_{2j}   $, proving the result. 
\item
Suppose $ u= t^*$.  By the inductive hypothesis, there is a set 
$\Lambda \subseteq \Delta \times \Delta $ such that
 $$ \bar{t} \;  \le_{EX}  \sum_{(\delta,\delta') \in \Lambda }  \phi( \delta) \parallel \phi(\delta').  $$ 
 Define  the set 
$$\Psi = \{ (\delta_{11} \ldots \delta_{b1}, \delta_{12} \ldots \delta_{b2}) \vert \, b \ge 0, \text{ each } (\delta_{j1} , \delta_{j2}) \in \Lambda \}$$
(note that $(1_\Delta, 1_\Delta) \in \Psi$). 
For each $ ( \epsilon_1, \epsilon_2)\in \Lambda$,  
\begin{align*}
 \phi( \epsilon_1) \parallel \phi( \epsilon_2)    & \sum_{( \delta_1,\delta_2)\in \Psi} \phi( \delta_1) \parallel \phi ( \delta_2) && \le_{EX}
  \\
&  \sum_{( \delta_1,\delta_2)\in \Psi} \phi( \epsilon_1)  \phi( \delta_1) \parallel \phi ( \epsilon_2) \phi( \delta_2)   && \le_{bw-Rat}
   \\
&   \sum_{( \delta_1,\delta_2)\in \Psi} \phi( \epsilon_1 \delta_1) \parallel \phi ( \epsilon_2 \delta_2) &&
\le_{bw-Rat} 
\\
& \sum_{( \delta_1,\delta_2)\in \Psi} \phi( \delta_1) \parallel \phi ( \delta_2)
\end{align*}
follows from the exchange law and Theorem \ref{thm.bwrat.equal.imply.bw-equiv} and the fact that $ (\delta_1,\delta_2)\in \Psi \Rightarrow ( \epsilon_1\delta_1, \epsilon_2 \delta_2)\in \Psi  $ always holds, 
 and hence 
\begin{align*}
 u = \bar{t^*}  \; & \le_{bw-Rat}    
   && \bar{t^*}        \sum_{( \delta_1,\delta_2)\in \Psi} \phi( \delta_1)     
      \parallel     \phi ( \delta_2)
 \\
 &\le_{EX}   && \big(        \sum_{(\epsilon_1,\epsilon_2)\in \Lambda} \phi( \epsilon_1)  \parallel \phi( \epsilon_2)   \big)^*  \sum_{( \delta_1,\delta_2)\in \Psi} \phi( \delta_1) \parallel \phi ( \delta_2)
 \\
 & \le_{EX}      &&    \sum_{( \delta_1,\delta_2)\in \Psi} \phi( \delta_1)  \parallel \phi ( \delta_2) 
  \end{align*}
  using the induction axiom of Kleene algebra. 
In addition, if $\lambda \in \Psi$, then
 $$ \lambda = (\delta_{11} \ldots \delta_{b1}, \delta_{12} \ldots \delta_{b2})$$
  for some  $ b \ge 0$, and  each $ (\delta_{j1} , \delta_{j2}) \in \Lambda $; and by the inductive hypothesis, for each $ j \le b$ and $ i \in \{ 1,2\}$  there is a word $w_j \in \lb t \rb $ such that the product of all the $i$th components of the letters of $w_j$ is $ \delta_{ji}$. Clearly the word  $ w=  w_1 \ldots w_b  \in \lb u \rb$, and the product of all the $i$th components of the letters of $w$ is $ \delta_{1i} \ldots \delta_{bi}$, completing the proof. \qed
\end{itemize}

\end{proof}

\begin{thm}[$ \odot$ preserves bw-rationality]
 \label{thm.odot.preserve.bwRat}
Let $r_1, r_2 $ be bw-rational terms over an alphabet. Then the language 
$ \lb  r_1 \rb  \odot \lb  r_2  \rb $ is bw-rational and satisfies $   \lb  r_1 \rb  \odot \lb   r_2  \rb    =_{EX} r_1 \parallel r_2  $. 
\end{thm}

\begin{proof}
Assume that each term $r_i \in T_{bw-Rat}( \Sigma)$ for  an alphabet $ \Sigma$. 
 We will first prove that 
 \begin{equation}
  \lb  r_1 \rb  \odot \lb   r_2  \rb    \le_{EX} r_1 \parallel r_2  \label{eqn.odot.EX-prec.para}
\end{equation}     by induction  on 
 $ \sum_{i = 1}^2 \width(r_i)$.  
 We will do this as follows. Clearly 
  $\lb  r_1 \rb  \odot \lb   r_2  \rb \subseteq \Sigma \cup \{ 1 \} \cup \seq \cup \, \cup_{k=1}^{\width  ( \lb  r_1 \rb  \odot \lb   r_2  \rb   )}  \para_k$. 
  For any $ k \ge 2$, it follows  from Corollary \ref{2058f.turkey.fun.cor} and the inductive hypothesis that  the language $    \lb  r_1 \rb  \odot \lb   r_2  \rb    \cap \para_k$ is bw-rational 
and satisfies  $   \lb  r_1 \rb  \odot \lb   r_2  \rb    \cap \para_k \,  \le_{EX} r_1 \parallel r_2  $.
The same assertion with  $\para_k$ replaced by $ \Sigma \cup \{ 1 \}$ obviously holds. 
Thus it suffices to prove that 
  $   \lb  r_1 \rb  \odot \lb   r_2  \rb    \cap  \seq $ is bw-rational and \begin{equation}
  \lb  r_1 \rb  \odot \lb   r_2  \rb    \cap  \seq  \, \le_{EX}  r_1 \parallel r_2 \label{eqn.r1.odot.r2.cap.EX}
\end{equation} 
  holds. 
This will be proved using Corollary \ref{cor.odot.into.reg.over.para}.

By Lemma \ref{rat.lang.jay249y.partition.lem} and Proposition \ref{devil.386f.kite.prop} and its sequential counterpart, and by ignoring the cases in their proofs that refer to parallel iteration $ ^{(*)}$, 
 there are 
  regular terms $t_i= t_i( \gamma_1, \ldots, \gamma_m)$ over an alphabet
 $ \Gamma = \{\gamma_1, \ldots, \gamma_m \}$ satisfying  $ r_i =_{bw-Rat}  t_i(c_1,\ldots, c_m)$ for non-sequential  terms  $c_1,\ldots, c_m$. 
By Lemma \ref{6299x.reg.lang.equiv.lem}, there exists a congruence $\approx$ of finite index of the monoid $( \Gamma^*  , 1, \cdot)$
 such that each language  $ \lb t_i \rb$ is the union of  a subcollection of $\approx$-equivalence classes. Define the languages $C_j = \lb c_j \rb$ for each $ j \le m$.

 Let $ \Delta$ be a set such that there is a bijection $ \phi$ from $ \Delta$ to the set of  $\approx$-equivalence classes. 
 Then by Corollary \ref{cor.odot.into.reg.over.para}, there is a regular term $u$  with $\supp(u) \subseteq \Delta \times \Delta $ such that
 (\ref{eqn.reg.over.odot.intersect.not_seq}) holds, 
 and for each word $ ( \delta_{11},  \delta_{21} ) \ldots  ( \delta_{1b},   \delta_{2b} ) \in \lb u \rb $ and  $ i =1,2$, (\ref{eqn.delta.without.C-j.into.t-i}) holds. 
 
Let $(\delta_1, \delta_2 )  \in \supp(u)$. 
 Then by  (\ref{eqn.width.phi.less.than.t-i}), 
 $ \width \big( \phi(\delta_i)( C_1, \ldots, C_m) \big) \le      \width (r_i)$ holds for each  $i =1,2$. Hence by  the inductive hypothesis, we can apply 
  Corollary \ref{2058f.turkey.fun.cor}  to the languages 
$   \phi(\delta_i)(C_1,\ldots, C_m)$,  
and so    for each $k \ge 2$, 
\begin{align*}
&   \phi(\delta_1)(C_1,\ldots, C_m)   \odot  \phi(\delta_2)(C_1,\ldots, C_m)  \,   \cap \para_k \\
 &\le_{EX}\\ 
&\phi(\delta_1)(C_1,\ldots, C_m)  \parallel \phi(\delta_2)(C_1,\ldots, C_m) 
\end{align*}
holds. 
In addition, the same statement with $ \para_k$ replaced by $ \Sigma$ is clearly true. Thus
by taking the union of the languages
 $ \phi(\delta_1)(C_1,\ldots, C_m)  \odot  \phi(\delta_2)(C_1,\ldots, C_m)  \big)    \cap \para_k$ 
for each
\\
 $k \in \{ 2, \ldots, \width\big( \phi(\delta_1)(C_1,\ldots, C_m)   \odot  \phi(\delta_2)(C_1,\ldots, C_m)  \big)   \}$ and also the language \\
  $ \phi(\delta_1)(C_1,\ldots, C_m)   \odot  \phi(\delta_2)(C_1,\ldots, C_m)  \,   \cap \Sigma$  it follows that 
\begin{align*}
& \phi(\delta_1)(C_1,\ldots, C_m)   \odot  \phi(\delta_2)(C_1,\ldots, C_m)  \,     \cap (\para \cup \Sigma)\\
 &\le_{EX}\\ 
&\phi(\delta_1)(C_1,\ldots, C_m)  \parallel \phi(\delta_2)(C_1,\ldots, C_m) 
\end{align*}
holds. 
Hence by (\ref{eqn.reg.over.odot.intersect.not_seq}),  the language  $ \lb r_1 \rb \odot \lb    r_2 \rb \, \cap \, \seq $
is bw-rational and 
\begin{align}
&\lb r_1 \rb \odot \lb    r_2 \rb \, \cap \, \seq \nonumber \\
& \le_{EX} \nonumber   \\
&  u\big( (\delta_1, \delta_2) \mathlarger{\setminus} \big(  \phi(\delta_1)(C_1,\ldots, C_m) \parallel  \phi(\delta_2)(C_1,\ldots, C_m)  \big)  \vert \; \delta_1, \delta_2 \in \Delta \big) \label{eqn.reg.paraprod.prec.parasum} 
 \end{align}
holds. 
By  Lemma \ref{lem.ordering.reg.over.paral}, there exists a set $ \Lambda \subseteq \Delta \times \Delta$ such that 
\begin{equation}
  u \big((\delta_1, \delta_2) \mathlarger{\setminus} \phi(\delta_1)  \parallel \phi(\delta_2) \vert \; \delta_1, \delta_2 \in \Delta  \big)  \le_{EX} \sum_{(\delta_1, \delta_2) \in \Lambda} \phi(\delta_1) \parallel \phi(\delta_2)     \label{eqn.reg.over.para.to.sumpara}
\end{equation}
and
 for each $ (\epsilon_1, \epsilon_2) \in \Lambda $, there is a word 
$(\delta_{11},  \delta_{21} ) \ldots   ( \delta_{1b},   \delta_{2b} ) \in \lb u \rb $ such that 
$$ \emptyset \not=   \phi(\delta_{i1}) \ldots \phi(\delta_{ib})  \subseteq  \phi( \epsilon_i) $$
 for each $ i \in \{1,2\}$, where if $b=0$ the product 
$ \phi(\delta_{i1}) \ldots \phi(\delta_{ib}) $ is defined to be the language $\{ 1 \}$;
and since each language  $\lb  {t}_i \rb$ is  a union of $\approx$-equivalence classes, 
$  \phi( \epsilon_i) \subseteq \lb  {t}_i \rb $ follows from (\ref{eqn.delta.without.C-j.into.t-i}) and so by Theorem \ref{thm.bwrat.equal.imply.bw-equiv}, 
$$ \sum_{(\epsilon_1, \epsilon_2) \in \Lambda} \phi(\epsilon_1)(C_1,\ldots, C_m) \parallel \phi(\epsilon_2)(C_1,\ldots, C_m) \, \le_{EX} r_1 \parallel r_2 $$
and so  (\ref{eqn.r1.odot.r2.cap.EX})   follows from (\ref{eqn.reg.paraprod.prec.parasum}) and (\ref{eqn.reg.over.para.to.sumpara}) with the languages $C_j$ substituted for $ \gamma_j$.

 Hence we have proved   (\ref{eqn.odot.EX-prec.para}).  Since   $ \,  \lb  r_1 \rb  \odot \lb   r_2  \rb    \supseteq \lb r_1 \parallel r_2 \rb $ clearly holds,  by Theorem \ref{thm.bwrat.equal.imply.bw-equiv} we can replace $ \le_{EX}$  by $ =_{EX}$ in (\ref{eqn.odot.EX-prec.para}). 
    \qed  

\end{proof}

Our main theorems concerning pomset ideals follow. 

\begin{thm}[$\id_{sp}$ preserves bw-rationality of pomset languages] \label{3696r.trex.ideal.rational.thm}
Let $t$ be a  bw-rational pomset term. Then $ \id_{sp}(\lb t\rb)$ is bw-rational and \\
$  \id_{sp} (\lb t\rb) \le_{EX}  t $ holds. 
\end{thm}

\begin{proof}
 The result  follows by induction on the structure of $t$. If $t \in \{0,1\}$ or $t$ is a letter, then the result is immediate. If $t = u + v$ then the result follows since $ \id_{sp} (\lb t\rb) =  \id_{sp} (\lb u\rb) +  \id_{sp} (\lb v\rb)$. 
If $t= u^*$ then the result follows since $ \id_{sp} (\lb t\rb) =( \id_{sp} (\lb u\rb))^*$. If $t= uv $ then  the result follows since  $ \id_{sp} ( \lb  t\rb) =  \id_{sp} (\lb u\rb) \id_{sp} ( \lb v\rb)$. Lastly if $t = r_1\parallel r_2  $ then the result follows from Theorem \ref{thm.odot.preserve.bwRat} since $  \id_{sp} (\lb t\rb) =  \id_{sp} ( \lb r_1 \rb) \odot  \id_{sp} (\lb r_2\rb)$  by (\ref{odot.eqn}) in Lemma \ref{oslash.852t.para.induct.lem}.
\qed
\end{proof}

\begin{thm}[free  bw-rational algebras with exchange law given as ideals] \label{3960h.mice.rational.exchange.derive.thm}
Let $ \Sigma$ be an alphabet and let  $a,b \in T_{bw-Rat}$. Suppose that $  \id_{sp} (\lb a \rb) =   \id_{sp} ( \lb b \rb)$ holds. Then $ a =_{EX} b $ holds. Thus the isomorphic bw-rational algebras
\begin{equation}
\big\{ \id_{sp}( \lb t \rb ) \big\vert \, t \in T_{bw-Rat}(\Sigma) \big\} \text{ and }  
\big\{ \id( \lb t \rb ) \big\vert \, t \in T_{bw-Rat}(\Sigma) \big\}  
\end{equation}
with $ \parallel$ interpreted as in (\ref{id.sp.interp.eqn}) and (\ref{id.interp.eqn}) respectively, 
are both freely generated in the class of bw-rational algebras satisfying (\ref{exch.eqn})
 by the elements $ \{ \sigma \}$ for $ \sigma \in \Sigma$. 
\end{thm}

\begin{proof}
By   Theorem \ref{3696r.trex.ideal.rational.thm},   the languages  $\id_{sp} ( \lb a \rb) $ and $  \id_{sp} ( \lb b \rb) $ are bw-rational, and so  we  have $$   a \le_{bw-Rat}  \id_{sp} ( \lb a \rb) \le_{bw-Rat}  \id_{sp} ( \lb b \rb) \le_{EX}  b, $$ 
where the first two relations follow from   Theorem \ref{thm.bwrat.equal.imply.bw-equiv}
and  the last relation follows from Theorem \ref{3696r.trex.ideal.rational.thm}. Interchanging $a$ and $b$ in this argument gives  $ a =_{EX} b $. 
The freeness assertion then follows from Theorem \ref{thm.id.subalgbra.iso.bi-Kleene}.
\qed
\end{proof}

\section{Conclusions}
\label{sect.conclusions}

We have proved, in this paper, that the class of pomset languages is closed under all Boolean operations, and that every identity that is valid for all pomset languages is a consequence of the set of valid regular and commutative-regular identities. We have also shown that the problem of establishing whether two pomset terms define the same language is decidable. The complexity of this is not clear however. 
  It is known that decidability of equivalence of two regular terms is PSPACE-complete \cite{Stockmeyer:Meyer:1973:word:exponential:problems,phd-stockmeyer:1974:Automata:Logic}, and can be shown that the analogous problem for commutative-regular terms lies in PSPACE, hence it is possible that generalising to pomset terms does not increase the bound beyond PSPACE. On the other hand, this problem may be EXPTIME-complete or EXPSPACE-complete. This is worth investigating further. 

\bibliographystyle{elsart-num}
\bibliography{pomset}

\end{document}